\documentclass[aps,a4paper,twocolumn,superscriptaddress]{revtex4-1}
 \pdfoutput=1

\usepackage[toc,page,titletoc]{appendix}
\usepackage{amsmath}
\usepackage{amssymb}
\usepackage{ amsthm }
\usepackage{braket}
\usepackage{amsbsy} 
\renewcommand\bra[1]{{\langle{#1}|}}

\makeatletter
\renewcommand\ket[1]{%
  \@ifnextchar\bra{\k@t{#1}\!}{\k@t{#1}}%
}
\newcommand\k@t[1]{{|{#1}\rangle}}
\newcommand{\singleket}[1]{\k@t{#1}}
\makeatother

\usepackage{graphicx}
\usepackage{natbib}
\usepackage{comment}

\usepackage[breaklinks=true]{hyperref}
\hypersetup{
  colorlinks   = true, 
  urlcolor     = blue, 
  linkcolor    = blue, 
  citecolor   = red 
}
\usepackage[nameinlink,capitalize]{cleveref}
 
\newcommand{\supplementary}{Appendix}

\newcommand\tr{\mathrm{Tr}}

\newcommand{\cc}[1]{{#1}^\ast}

\newcommand{\abs}[1]{\left\vert{} #1 \right\vert}
\newcommand{\norm}[1]{\left\Vert{} #1 \right\Vert}
\newcommand{\ii}{\mathrm{i}}
\newcommand{\adj}[1]{{#1}^\dagger}
\renewcommand{\vec}[1]{\boldsymbol{#1}}
\newcommand{\mat}[1]{\boldsymbol{#1}}

\newcommand{\BC}{{\mathbb{C}}}
\newcommand{\CO}{{\mathcal{O}}}

\newcommand{\Hermitian}{\mathbb{H}}
\newcommand{\Sphere}{\mathcal{S}}
\DeclareMathOperator*{\argmin}{argmin}

\newtheorem{theorem}{Theorem}

\newtheorem{corollary}[theorem]{Corollary}
\newtheorem{lemma}[theorem]{Lemma}
\newtheorem{proposition}[theorem]{Proposition}

\newtheorem{protocol}[theorem]{Protocol}

\newcommand{\Cologne}{Institute for Theoretical Physics, University of Cologne, Germany}
\newcommand{\Bristol}{Quantum Engineering and Technology Laboratories, School of Physics and Department of Electrical and Electronic Engineering, University of Bristol, UK}
\newcommand{\Imperial}{Department of Physics, Imperial College London, London, UK}

\newcommand{\Kepler}{Institute for Integrated Circuits, Johannes Kepler University Linz, Altenbergerstrasse 69, 4040 Linz, Austria}
\newcommand{\NobuB}{Currently at: Department of Communications Engineering, Graduate School of Engineering, Tohoku University, Sendai, Japan}
\newcommand{\NobuA}{NTT Basic Research Laboratories, NTT Corporation, Atsugi, Japan}
\newcommand{\hashi}{  NTT Device Technology Laboratories, NTT Corporation, Atsugi, Japan}

\begin{document}
 
\title{Rapid characterisation of linear-optical networks via PhaseLift}
\author{D.\ Suess}
\affiliation{\Cologne}
\author{N.\ Maraviglia}
\affiliation{\Bristol}
\author{R.\ Kueng}
\affiliation{\Kepler}
\author{A.\ Ma\"inos}
\affiliation{\Bristol}
\author{C.\ Sparrow}
\affiliation{\Bristol}
\affiliation{\Imperial}
\author{T.\ Hashimoto}
\affiliation{\hashi}
\author{N.\ Matsuda}
\affiliation{\NobuA}
\affiliation{\NobuB}
\author{D.\ Gross}
\affiliation{\Cologne}
\author{A.\ Laing}
\affiliation{\Bristol}

\begin{abstract}
  Linear-optical circuits are elementary building blocks for classical and quantum information processing with light.
  In particular, due to its monolithic structure, integrated photonics offers great phase-stability and can rely on the large scale manufacturability provided by the semiconductor industry.
  New devices, based on such optical circuits, hold the promise of faster and energy-efficient computations in machine learning applications and even implementing quantum algorithms intractable for classical computers.
  However, this technological revolution requires accurate and scalable certification protocols for  devices that can be comprised of thousands of optical modes. 
  Here, we present a novel technique to reconstruct the transfer matrix of linear optical networks that is based on the recent advances in low-rank matrix recovery and convex optimisation problems known as PhaseLift algorithms. Conveniently, our characterisation protocol can be performed with a coherent classical light source and photodiodes. 
  We prove that this method is robust to noise and scales efficiently with the number of modes.
  We experimentally tested the proposed characterisation protocol on a programmable integrated interferometer designed for quantum information processing. 
  %
  We compared the transfer matrix reconstruction obtained with our method against the one provided by a more demanding reconstruction scheme based on two-photon quantum interference.
  %
  For 5-dimensional random unitaries, the average circuit fidelity between the matrices obtained from the two reconstructions is $0.993$.
\end{abstract}

\maketitle
\section*{Introduction}
\label{sec:introduction}

\subsection*{Motivation}

Linear optical networks (LON) are fundamental to the processing of quantum and classical information with light. 
Passive and reconfigurable linear optical circuits have been proposed and demonstrated for many applications
including telecommunications~\cite{Miller2015-sortingBeams}, as processing units for machine learning~\cite{Vandoorne2014-reservoir,Shen2017-deepLearning,Lin2018-DiffractiveNN,Roques2020-Ising,Abel2019-nuromorphic}, and as platform for quantum computation and simulation~\cite{Wang2019-bosonsampling,Asavanant2019-ClusterState,Sparrow2018-molecule}.
With the continuing development of programmable large-scale integrated photonic platforms~\cite{Wang2018-16D,Taballione2019-8x8SiN,Seok2016-switch,Perez2017-fieldProgrammable,Chung2018-PhasedArray}, practical and reliable techniques for characterising and validating the operation of these devices are crucial.
In this work, we present a new protocol for characterising linear optical devices with low experimental resources by expressing the relation between measured intensities and linear properties of LONs as a phase retrieval problem \cite{walther_question_1963}.


A phase-stable LON is characterised by its complex transfer matrix $\mat{M}$. 
The amplitudes of the output light modes, $\beta_j$, depend on the amplitudes of the input modes, $\alpha_k$, via 
\begin{equation}
  \beta_j = \sum_k M_{jk} \,\alpha_k.
  \label{eq:coherent_transfer_matrix}
\end{equation}
Arguably, determining $\mat{M}$ experimentally is a crucial step to validate and verify an existing LON.

To characterise phase-stable LON, several protocols that feature alternative reconstruction and optimisation algorithms have been demonstrated  \cite{laing2012-superstable,rahimi2013-DirectCharacterisation,Heilmann2015-characterisationImprovement,Spagnolo2017-genetic,Tillmann2016-reconstruction}. With the exception of the work presented in \cite{rahimi2013-DirectCharacterisation}, further developed in \cite{Heilmann2015-characterisationImprovement}, all these schemes rely upon non-classical two-photon interference measurements. 
Instead, one of the great advantages of the protocol discussed in  \cite{rahimi2013-DirectCharacterisation} is that it can be performed with a classical coherent light source and power-meters. However, this benefit can be hindered by the lack of statistical inference of the matrix elements from an over-complete set of data that would instead compensate for experimental sources of errors \cite{Tillmann2016-reconstruction}.  
Linear optical circuits are also used in the context of quantum computation to implement quantum gates. Characterising these quantum transformations requires the use of quantum process tomography \cite{OBrien2004-ProcessTomography,Rahimi2011-thomography}, even if implemented by a linear optical systems. In particular, an $n$ modes linear optical system can be treated  as an $n-$dimensional qudit channel for a single photon state \cite{Varga2018-QuditChannel}. Resorting to these approaches generally requires the capability of preparing and detecting quantum states of light and the acquisition of larger datasets. In return, they can provide additional information on the noise affecting the quantum system due, for example, to incoherent scattering. 

In this work, we demonstrate a reconstruction procedure based on efficient optimisation algorithms designed to be resilient to experimental imperfection and that can be performed with classical instrumentation, i.e. a coherent light source and power-meters. 
%

\subsection*{Background: the PhaseLift algorithm}
\label{sec:phase_retrieval}
In classical optical experiments, the standard measurable quantities are the intensities, or the power, of the output modes
\begin{equation}
  I_j(\vec{\alpha})
  = \left| \beta_j \right|^2 + \epsilon_j
  = \left| \sum_k M_{jk} \, \alpha_k \right|^2 + \epsilon_j,
  \label{eq:intensities}
\end{equation}
for certain coherent input patterns with amplitudes $\vec{\alpha}=\{\alpha_1, \dots, \alpha_n\}$.
Here, $\epsilon_j$ describes noise due to statistical fluctuations or systematic errors.
Although the output states~\eqref{eq:coherent_transfer_matrix} are linear in $\mat{M}$, the resulting intensity measurements~\eqref{eq:intensities} are quadratic in $\mat{M}$ and oblivious to the phases of $\beta_j$.
In particular, the problem of reconstructing the matrix $\mat{M}$ from such a set of data is ill-posed since all the measured intensities are invariant under the multiplication of any row of the matrix by an arbitrary phase-factor $e^{i \phi_j}$. 

The crucial observation in this paper is that measurements~\eqref{eq:intensities} closely resemble the model of the \textit{phase retrieval problem}, i.e.\ the problem of recovering a complex vector $\vec{x} \in \BC^n$ from $m$ scalar measurements of the form
\begin{equation}
  y^{(l)} = \left| \langle \vec{x}, \vec\alpha^{(l)} \rangle \right|^2 + \epsilon^{(l)}
  \quad l=1,\ldots,m.
  \label{eq:phase_retrieval_measurements}
\end{equation}
Here, $\vec\alpha^{(l)} \in \BC^n$ denote measurement vectors and $\epsilon^{(l)}$ the additive measurement errors.
One practical solution to the phase retrieval problem~\cite{Balan2009-FrameCoefficients} -- and, by extension, for recovering transfer matrices -- is based on its connection to the field of low-rank matrix recovery~\cite{Ahmed2014-ConvexDeconvolution,Candes2009-MatrixCompletion,Candes2011-OracleBounds,recht2010-LinearEqautionSolutions,Gross2011-LowRankRecovering,chen2013-MatrixCompletion}. 

Note that the measurements \eqref{eq:phase_retrieval_measurements} are quadratic in the target vector $\vec{x} \in \mathbb{C}^n$, but \emph{linear} in its outer product $\ket{\vec{x}} \bra{\vec{x}} \in \mathbb{C}^{n \times n}$:
\begin{equation}
  \left| \langle \vec{x}, \vec\alpha^{(l)} \rangle \right|^2 + \epsilon^{(l)}
  = \tr \left( (\ket{\vec\alpha^{(l)}}\bra{\vec\alpha^{(l)}}) (\ket{\vec{x}}\bra{\vec{x}}) \right)+ \epsilon^{(l)}. \label{eq:matrix-measurements}
\end{equation}
This ``lifts'' the phase retrieval problem to the problem of recovering  $\mat{X}=\ket{\vec{x}} \bra{\vec{x}}$ from linear measurements.
This target matrix has rank one, $\mathrm{rank}(\mat{X})=1$, and is also positive semidefinite (psd), $\mat{X}\geq 0$. 
The connection to low-rank matrix recovery is now apparent. We need to find the lowest-rank matrix $\mat{X}\geq 0$ that is  compatible with the measurement data. This can be done with an algorithm known as \emph{PhaseLift} \cite{Candes2013_Phaselift}:
\begin{align}
\underset{\mat{Z} \in \mathbb{C}^{n \times n}}{\text{minimize}} & \quad \tr ( \mat{Z}) \\
\text{subject to} & \quad \sum_{l=1}^m \left| \tr \left( \ket{\vec\alpha^{(l)}} \bra{\vec\alpha^{(l)}} \, \mat{Z} \right) - y^{(l)} \right|^2 \leq \eta, \nonumber \\
& \quad \mat{Z} \geq 0.  \nonumber
\end{align}
Here, $\eta \geq \sum_{l=1}^m (\epsilon^{(l)})^2$ is an upper bound on the noise strength and the trace, $\tr (\mat{X})$, penalizes rank among psd matrices \cite{Ahmed2014-ConvexDeconvolution,Candes2009-MatrixCompletion,Candes2011-OracleBounds,recht2010-LinearEqautionSolutions,Gross2011-LowRankRecovering,chen2013-MatrixCompletion}. 
In this work, we will use a variant of PhaseLift  that does not require any prior knowledge about the noise strength \cite{kabanava_stable_2016}. 
Instead, 
we can directly minimize a simple loss function over the set of psd matrices:
\begin{align}
  \underset{\mat{Z}\in \mathbb{C}^{n \times n}}{\textrm{minimize}} & \quad \sum_{l=1}^m \left| \tr \left( \ket{\vec\alpha^{(l)}} \bra{\vec\alpha^{(l)}} \, \mat{Z} \right) - y^{(l)} \right| \label{eq:PhaseLift} \\
  \textrm{subject to} &\quad  \mat{Z} \geq 0. \nonumber
\end{align}
Here, we have chosen the $\ell_1$-loss function which is known to be exceptionally robust with respect to noise corruptions $\epsilon^{(1)},\ldots,\epsilon^{(m)}$ \cite{kabanava_stable_2016}. 
 The more commonly used least-squares loss function would also produce qualitatively similar results. 
 
The minimizer $\mat{Z}^\sharp$ of Algorithm~\eqref{eq:PhaseLift}
is a psd matrix and must be factorized to recover the estimated signal vector $\vec{x}^\sharp \in \mathbb{C}^n$. After applying an eigenvalue decomposition to $\mat{Z}^\sharp$, we set $\vec{x}^\sharp$ to be the eigenvector associated with the largest eigenvalue $\lambda$, re-scaled to length $\norm{\vec{x}^\sharp}=\sqrt{\lambda}$ \cite{Candes2013_Phaselift}. Note that $\vec{x}^\sharp$ is only recovered up to an arbitrary phase factor $e^{i\phi}$ -- an unavoidable ambiguity for the phase retrieval problem \eqref{eq:phase_retrieval_measurements}.


The PhaseLift algorithm (and its variants) belongs to a subclass of convex optimisation problems called semidefinite programs. Indeed, Algorithm \eqref{eq:PhaseLift} minimizes a convex loss function over the convex set of psd matrices. Such optimisation problems have no local optima (or saddle points) except for the global optimum that is essentially unique. 
Dimensions $n$ of many thousands can be handled by scalable semidefinite programming algorithms~\cite{BM03:Nonlinear-Programming,BBV16:Low-Rank-Approach,YTF+19:Scalable-SDP}.

Reliability and (to some extent) scalability are key advantages of phase retrieval via PhaseLift over alternative optimisation approaches that do not rely on lifting, see e.g.\ \cite{YUTC17:Sketchy-Decisions}.
Minimizing a loss function directly over vectors $\vec{x} \in \mathbb{C}^n$
results in an optimisation problem that is lower-dimensional, but not convex. 

Phase retrieval via PhaseLift is not only a compelling heuristic, it is also supported by rigorous theory. By and large, the theoretical
guarantees require stochastic generative models for the measurement vectors, i.e.\ each $\ket{\vec{\alpha}^{(l)}}$ in \cref{eq:phase_retrieval_measurements} is sampled independently from a suitable measurement ensemble.
A prominent example is the uniform (Gaussian) measurement ensemble \cite{Candes2013_Phaselift,Kueng2017-MatrixRecovery,kabanava_stable_2016}. 
However, ensembles that feature less randomness~\cite{Gross2015-Derandomization,Kueng2017-MatrixRecovery,kueng_low_2016} or additional structure tailored to specific applications~\cite{Candes2015-DiffractionPattern,Gross2017-ImprovedGuarantees,voroninski_quantum_2013,Kueng2015-MatrixRicovery} have also been investigated. 
The strongest theoretical performance guarantees assume the following form:

\begin{theorem} \label{thm:phaselift-intro}
Suppose that each phaseless measurement \cref{eq:phase_retrieval_measurements} $\vec{\alpha}^{(l)}$ is chosen uniformly at random from a suitable ensemble. Then, an order of $m = \mathrm{const} \times n$ measurements suffice to recover any complex vector $\vec{x} \in \mathbb{C}^n$ via PhaseLift.
More precisely, the solution $\mat{Z}^\sharp$ to \cref{eq:PhaseLift} obeys $  \norm{\mat{Z}^\sharp -  \ket{\vec{x}}\bra{\vec{x}}}_2 \propto \epsilon^{\text{tot}}/m$, where $\epsilon^{\mathrm{tot}}= \sum_{j=1}^m |\epsilon_j|$ is the total noise corruption in the measurement process.
\end{theorem}

The implication of such a result is twofold. First, it ensures that the number of measurements must scale linearly in the problem dimension $n$. Unfortunately, these theoretical results are ill-equipped to produce the exact proportionality constant. It is known that roughly $m=4n-4$ measurements are necessary ($\mathrm{const} \geq 4$) to solve the phase retrieval problem unambiguously \cite{heinosaari_quantum_2013}.
%
Second, the reconstruction is stable with respect to noise in the measurements  \cref{eq:phase_retrieval_measurements}. Accurate phaseless measurements produce accurate solutions $\vec{x}^\sharp \in \mathbb{C}^n$ to the phase retrieval problem. 

The main theoretical contribution of this work is a recovery guarantee -- similar to \cref{thm:phaselift-intro} -- for a novel measurement ensemble: the \emph{randomly erased complex Rademacher} (RECR) ensemble. 
Each measurement vector 
has random coefficients $\alpha_j$ that can take five distinct values: $0$ and $\pm1,\pm i$. We refer to \cref{eq:recr_definition} below for details.

\section*{LON Phaselift reconstruction}%
\label{sec:characterization}

\begin{figure*}[bt!] 
    \includegraphics[width=0.90\textwidth]{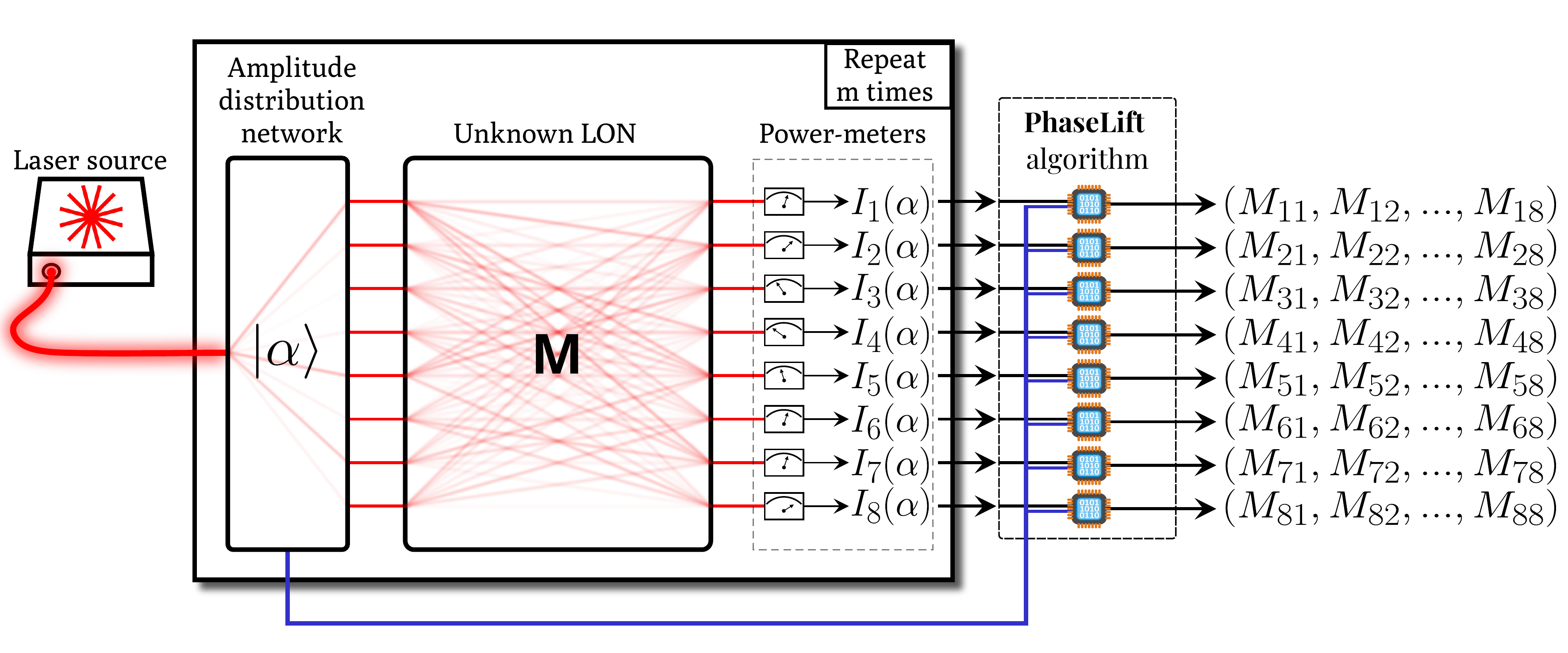} 
  \caption{%
    Schematic of the PhaseLift characterisation protocol. (see \cref{prot:characterization}).
    A sequence of input patterns $\singleket{\boldsymbol{\alpha}}$ are randomly sampled from the uniform or RECR ensembles.
    Such input vectors are implemented by a trusted distribution network that coherently distributes the amplitude of a laser source proportionally to the complex components of $\singleket{\boldsymbol{\alpha}}$. 
    The modulated light is injected into the input modes of the unknown linear optical network (LON) and the power at each output port is measured.
    For each output mode, the list of $\singleket{\boldsymbol{\alpha}}$ and the measurement outcomes of the power-meter are passed to the convex optimisation algorithm named PhaseLift that retrieves the matrix row associated to that specific output mode.
    Row by row, the unknown transfer matrix $\mat M$ that characterises the LON is then obtained up to a multiplicative phase factor for each row vector.
  }\label{fig:phaselift_protocol}
\end{figure*}

Let us now turn to connecting the two problems introduced in the last section, namely determining the transfer matrix $\mat{M}$ of a linear optical device on the one hand and the phase retrieval problem on the other hand.
Note that the measured intensity at detector $j$, as given by \cref{eq:intensities}, exclusively provides us with information about the $j$-th row vector of $\mat{M}$, $1 \leq j \leq n$:
\begin{equation}
  I_j(\vec{\alpha})
  = \left| \sum_{k=1}^n M_{jk} \alpha_k \right|^2 + \epsilon_j
  = \left\vert  \langle \vec{M}_j, \vec\alpha \rangle  \right\vert^2 + \epsilon_j. 
  \label{eq:intensities_as_overlap}
\end{equation}
Here, we have defined $\vec{M}_j$ as the (complex conjugated) row vectors of $\mat{M}$.
Since the measured intensities in \cref{eq:intensities_as_overlap} exactly resemble the measurement model of the phase retrieval problem in \cref{eq:phase_retrieval_measurements}, we can use the ideas mentioned in the introduction to reconstruct the transfer matrix. In particular, each projective measurement associated with the vector $\ket{\vec{\alpha}^{(l)}}$ corresponds to a power reading in a single output mode $j$ while light amplitudes proportional to the components of $\vec{\alpha}^{(l)}$ are injected into the input modes of LON.  
Therefore, we propose the following protocol, diagrammatically represented in \cref{fig:phaselift_protocol}.
\begin{protocol}{(for recovering the transfer matrix $\mat M$)}%
  \label{prot:characterization}
  \begin{enumerate}
    \item Sample $m$ random input states $\ket{\vec \alpha^{(l)}}$ from an appropriate ensemble.
    \item Measure the $m \times n$ intensities $I_1(\vec\alpha^{(l)}), \ldots, I_n (\vec \alpha^{(l)})$ with $l=1,\ldots,m$.
    \item Use PhaseLift~\eqref{eq:PhaseLift} to recover each $\vec{M}_j$ individually.
  \end{enumerate}
\end{protocol}
This protocol is able to reproduce transfer matrices without unitary assumptions and is suitable for non-squared  matrices too. In principle, with sufficiently precise measurements, this technique permits to quantify the degree of deviation from ideally unitary transformations.
The availability of a detector at each output mode facilitates a rapid reconstruction of the matrix since the same sequence of input vector $\vec{\alpha}^{(l)}$ can be used to independently recover multiple  rows of the matrix. 

Note that, to measure the intensities $I_j(\vec\alpha^{(l)})$, coherent light with amplitude proportional to $\alpha^{(l)}_k$ needs to be simultaneously injected in each input port $k$. Similarly to \cite{rahimi2013-DirectCharacterisation},  this procedure can be  performed with a single laser connected to the LON by means of a programmable phase-stable amplitude distribution network.  Additionally, in our reconstruction, a previous characterisation of the distribution network is usually required. However, the resilience to experimental errors of our method, based on the recovery guarantee characteristic of the Phaselift algorithm, can compensate for potential errors introduced by the preparation of the input vectors themselves. 

Two important questions remain:
(i) from which ensemble should we sample the input coherent states and (ii) how many such inputs are sufficient for a successful reconstruction?
In this section we provide two different answers to these questions.
First, we show that the established uniform measurement ensemble~\cite{Kueng2017-MatrixRecovery} allows for reconstructing $\mat M$ from an asymptotically optimal number of measurements.
Second, we show that, although it only requires a simplified light distribution network, the  RECR ensemble performs nearly as well as the uniform ensemble.  

\subsection*{Measurement ensembles}

\paragraph*{Uniform ensemble.}
The uniform sampling scheme consists of choosing $\vec \alpha$ uniformly from the complex unit sphere.
Up to normalisation, this is equivalent to choosing the real and imaginary part of the components of $\vec\alpha^{(l)}$ to be centred Gaussian random variables with variance $\tfrac{1}{2}$ each.
Fixing the norm of the input vectors to a constant is convenient for our particular application as it amounts to using the same input power for all the configurations $\vec\alpha^{(l)}$ and, therefore, simplifies the preparation procedure via unitary distribution networks.
%
Strong analytic reconstruction guarantees exist for phase retrieval with this measurement ensemble~\cite{candes_solving_2012, tropp_convex_2015, Kueng2017-MatrixRecovery}.
We provide a specific formulation for the problem at hand and a simplified proof strategy in the \supplementary.

\paragraph*{RECR ensemble.}
The uniform sampling scheme places high demands on the experimental implementation since it necessitates the ability to prepare any multi-mode coherent input state $\ket{\vec \alpha}$ with $\vec\alpha$ from the complex unit sphere.
Therefore, we propose an alternative measurement ensemble that lends itself to implementations in linear optics:
For $p \in (0,1)$, we define a \emph{randomly erased complex Rademacher} (RECR) random variable $a$ to be distributed according to
\begin{equation}
  a \sim
  \begin{cases}
    \pm 1, \pm \ii & \textrm{each with prob. } p/4 \\
    0 & \textrm{with prob. } 1 - p.
  \end{cases}
  \label{eq:recr_definition}
\end{equation}
For the RECR measurement model, we sample the components $\alpha_k$ of the input state $\ket{\vec \alpha}$ according to \cref{eq:recr_definition}, but we can additionally choose to normalise the total intensity, $\norm{\vec\alpha} = 1$.
Notably, a programmable optical circuit able to generate light amplitudes proportional to RECR vectors requires to set as little as four alternative phases and two intensity levels at each input mode.
We envisage that photonic devices can be optimised to implement such discrete input configurations; even phase-shifting strategies that are not tunable across a continuous range of phase shifts can be advantageously used to this scope \cite{Henriksson2018-DigitalMEMS,Dhingra2019-PhasechangingShifting}. 

\paragraph*{Performance guarantees.}

One important theoretical contribution of this work is to provide a rigorous proof of convergence for the proposed reconstruction scheme, which we outline now.
We refer the reader to the \supplementary\ for an exact formulation and the proofs.
\begin{theorem}[Informal version]%
  \label{thm:performance_guarantee}
  Suppose that $m \geq Cn$ input states have been chosen from either the uniform or the RECR ensemble.
  Then, with high probability, any transfer matrix $\mat{M}$ can be reconstructed via \cref{prot:characterization}.
\end{theorem}

This statement is to be understood as a theoretical performance guarantee in terms of an upper bound on the reconstruction error
\(
  \min_{\vec{\mu}}\left\| \mat{M}^\sharp -  \vec{D} (\vec{\mu}) \mat{M} \right\|_2.
\)
Here, $\mat{M}^\sharp$ is the reconstruction and $\mat{D}(\vec\mu) = \mathrm{diag}(\mu_1, \ldots, \mu_n)$ with $\abs{\mu_j} = 1$ are the row-phases of $\mat{M}$ that cannot be  recovered from the measurements~\eqref{eq:intensities}.
Our proofs do not give a tight bound for the constant $C$.
This is why we run numerical simulations in the following section in order to determine a practical value of $C$ which is conjectured to be 4.

\section*{Numerical analysis}
\begin{figure*}[hbt]
    \includegraphics[width=0.95\textwidth]{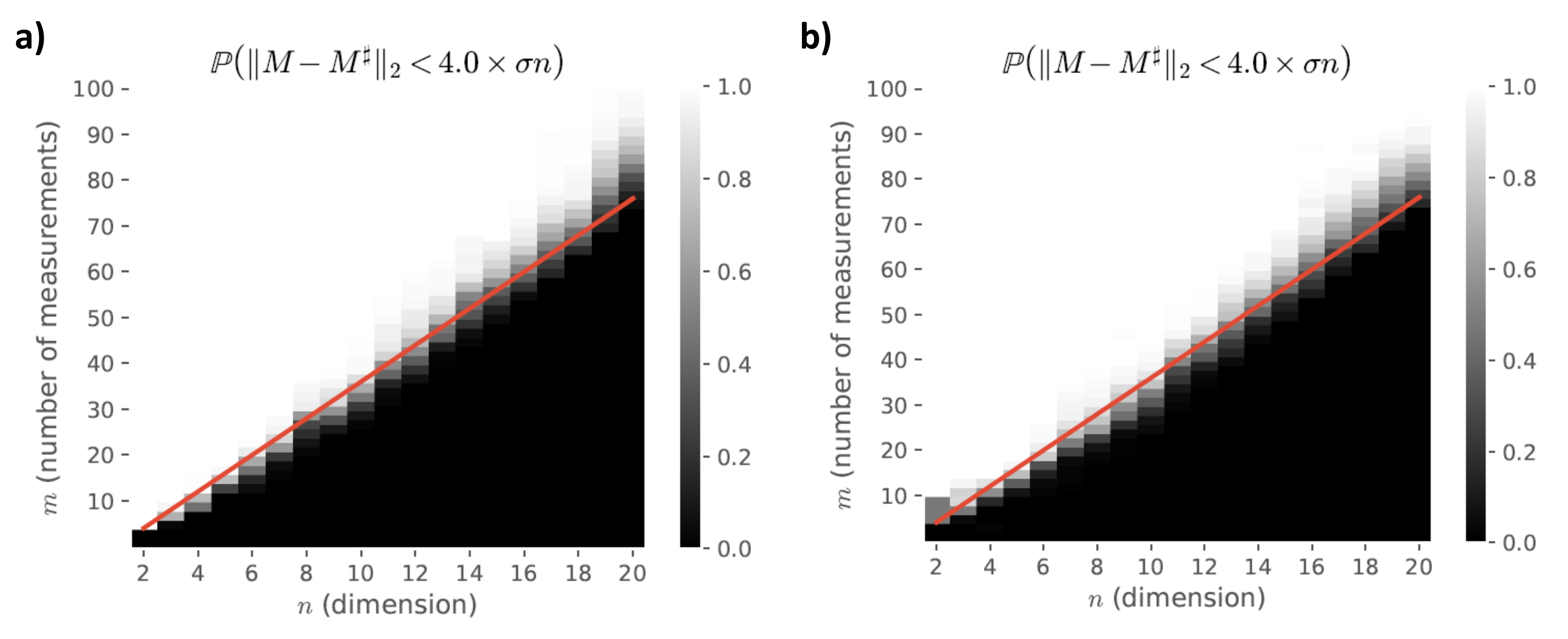}
  \caption{%
    Simulated probability for correctly recovering transfer matrices $\mat{M}$ using the two different sampling schemes under noisy measurements with $\sigma = 0.05$.
    The $x$-axis labels the problem dimension $n$, while the $y$-axis depicts different values for the number of (random) measurements $m$. 
    For each pair $(n,m)$, the (approximate) probability of correct transfer matrix reconstruction appears color-coded from black (zero) to white (one). Each probability is approximated by testing the protocol for 97 (Haar) random choices for $\mat{M}$, as well as the identity, the swap matrix and the discrete Fourier transform.
    For both uniform (a) and RECR (b) sampling,
    the probability of successful reconstruction undergoes a sharp phase transition just above $m=4n-4$ (red line).
  }\label{fig:simplot}\label{sfig:simplot.gaussian}\label{sfig:simplot.recr}
\end{figure*}

Firstly, we investigate the applicability of the PhaseLift characterisation protocol via numerical simulations.
The results depicted in \cref{fig:simplot} aim to visualise the performance guarantees from \cref{thm:performance_guarantee}:
For each given dimension $n$, we choose 100 target unitaries.
Each of these is reconstructed by means of \cref{prot:characterization} with a varying number of measurements $m$.
The input vectors are sampled from the uniform ensemble in \cref{sfig:simplot.gaussian}(a) and from the RECR ensemble in \cref{sfig:simplot.recr}(b).
For the measurement noise $\epsilon_j$ from \cref{eq:intensities_as_overlap}, we assume independent, centred Gaussian noise with standard deviation $\sigma = 0.05$.
The density plots show the fraction of successfully recovered unitaries.
Here, the criterion for success is whether the distance of the reconstruction $\mat{M}^\sharp$ measured in Frobenius norm is smaller than the threshold  $4 \sigma n$ in accordance with the error bound from \cref{cor:noisy_performance_guarantee} in the \supplementary. The two plots highlight that a sharp phase transition occurs just above the (red) line $m=4n-4$. The probability for correctly recovering $\mat{M}$ from $m$ uniform (left) and RECR (right) measurements jumps from zero (black color) to almost one (white color).
%
This demonstrates the high efficiency of \cref{prot:characterization} 
with respect to the number of measurements.
Not only does the number of measurements scale linearly in the system size but the proportionality constant is small as well. Hence, the PhaseLift algorithm is a practical candidate for characterising large-scale LONs.

\section*{Experimental Results}%
\label{sec:experiment}

\begin{figure*}[tbhp]
  \centering 
  \includegraphics[width=\textwidth]{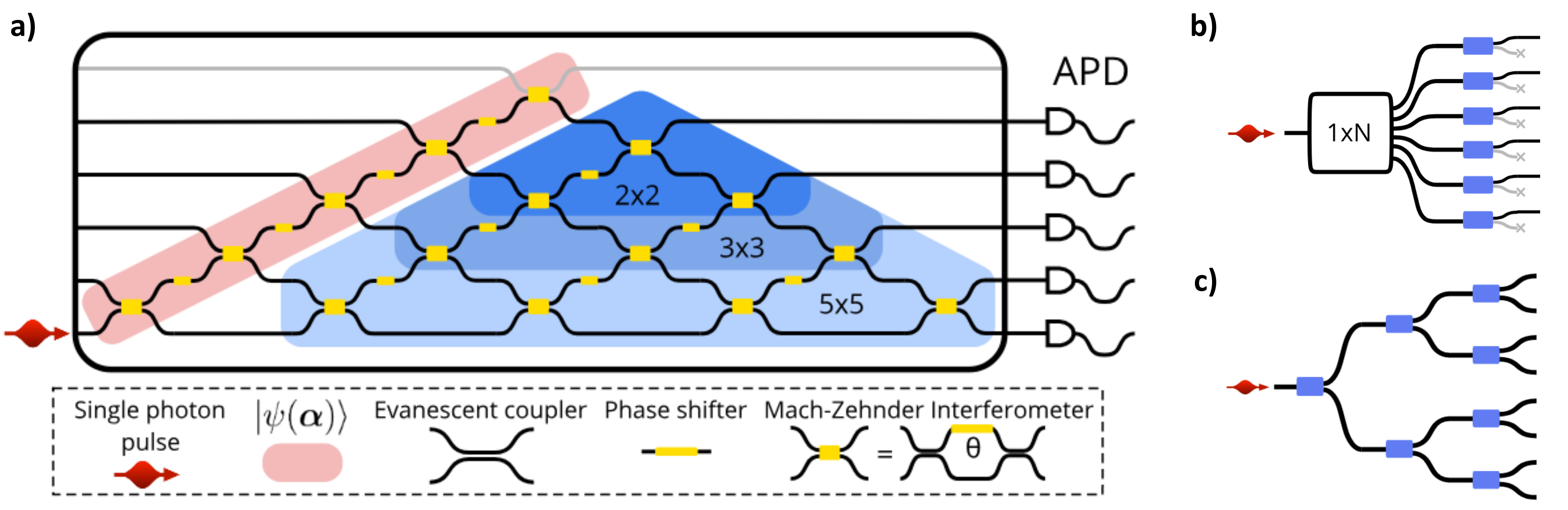}
  \caption{ \label{fig:experimental.schematic}
     a) Chip schematics of our experiment. Heralded single photons are injected into the bottom waveguide of a six-mode integrated interferometer.
     A diagonal sequence of Mach-Zehnder interferometers is used to prepare single-photon states $\singleket{\psi(\vec \alpha)} = \sum_k \alpha_k a_k^{\dag}   |\mathfrak{0}\rangle$ over the bottom five modes of the device. The remainder of the device is used to implement 2-, 3- and 5-dimensional unitary transformations which are to be characterised. 
    Each output port is coupled to an avalanche single photon detector (APD).
    b) and c) Alternative optical wiring for integrated devices to implement the amplitude distribution network. Blue squares represent generic tunable component for phase shifting and amplitude modulation.   b) broadcast and modulate, c) tree-like structure.
  }  
\end{figure*}




To experimentally verify our algorithm, we reconstructed multiple transfer matrices implemented by a reconfigurable integrated LON that has already been tested for quantum information processing \cite{Carolan2015-Universal}.
The device is comprised of 30 evanescent couplers and 30  thermo-optic phase-shifters acting on the fundamental optical modes of six single-mode waveguides. 
The schematic of the LON is shown in \cref{fig:experimental.schematic}(a). By injecting light into the bottom waveguide of the device, an initial sequence of five integrated Mach-Zehnder interferometers and five additional phase-shifters act as a distribution network to prepare the input vectors $\ket{\vec{\alpha}}$.
The remaining triangular mesh of components is then sufficient to implement the unitary transfer matrices $\mat{M}$ that we analysed~\cite{Reck1994}.
We note, however, that although the  design of the distribution network we used is sufficient to perform the PhaseLift reconstruction, it is not optimal since it does not minimise the average number of components the light goes through. 
For optimal performance, we suggest a distribution network design based on a tree-like connectivity or on a broadcast and modulate approach, see \cref{fig:experimental.schematic}(b,c) for schematic examples.

The LON was configured to implement several 2-, 3- and 5-dimensional $\mat{M}$, including identity and Fourier transformations, as well as uniformly (Haar) random unitaries. To test the quality of the PhaseLift reconstruction, we also performed an alternative reconstruction procedure for the same matrices based on a different physical principle: the interference signal in the second order correlation function as opposed to first order correlations \cite{Hong1987-HOMdip}. Indeed, in this second reconstruction algorithm, the phases of the matrix elements were inferred from two-photon interference measurements. In the \supplementary, we report more details on this method based on the work in \cite{laing2012-superstable,Tillmann2016-reconstruction} and performed by using photon pairs generated by a spontaneous down-conversion source. 

\begin{figure*}[tbp]
  \centering
  \includegraphics[width=\textwidth]{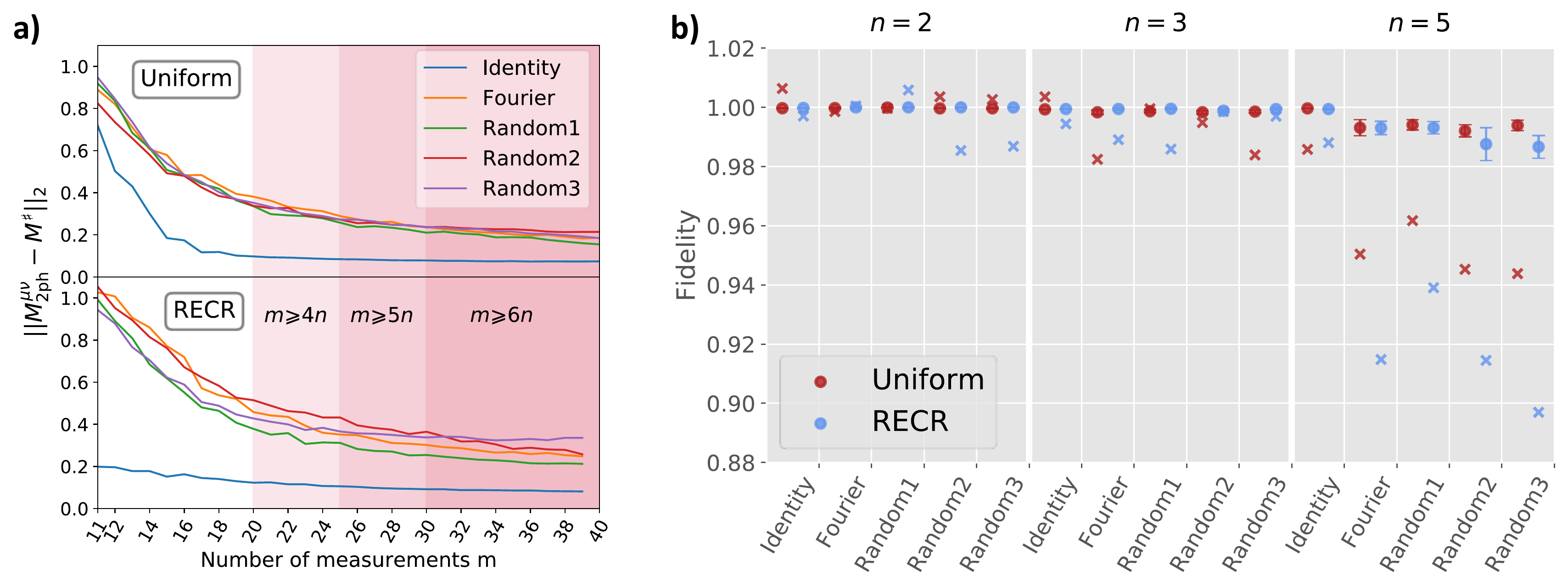}
  \caption{%
     \label{fig:experimental.phase_transition}
     \label{fig:experimental.overview}
     a) Experimental reconstruction distance, expressed with the Frobenius norm, between the PhaseLift reconstruction and the two-photon reference as a function of the number of input states $m$. Whenever possible, the distance is averaged over multiple combinations of the $m$ input vectors out of the measured available ones.   
     Performance of the uniform and RECR ensembles are similar. For non sparse unitaries, we witness convergence in line with the aforementioned conjecture m=4n. Their accuracy improves fast with the number of inputs before $m\sim 4n$ when it slows down.
     b) Circuit fidelity comparison for different dimensions of the target matrices. To both PhaseLift and two-photon reference reconstructions we apply the polar decomposition in order to compare unitary matrices. 
     For each matrix and sampling ensemble, we subsample 100 times $m = 6n$ input vectors out of the measured ones together with the corresponding measured intensities. 
     Error bars indicate the standard deviation of the distribution obtained from this re-sampling.     
     However, since for $n=2$ there are only six distinct RECR vectors up to a global phase, there is only one reconstruction.
     The crosses show the average fidelity that is observed without applying the polar decomposition to the  reconstructions. 
     }
\end{figure*}

Since the properties of the photonic components are generally mode-dependent, that is dependent on the wavelength, temporal envelope, polarisation, etc.\ of the light, we decided to perform the PhaseLift algorithm adopting the closest light source to the bi-photon states used for the two-photon reference reconstruction. Therefore, we used single photons generated by the same down-conversion source heralded via the detection of the co-generated twin photon in the idler mode. The opportunity of using single photons as alternative to coherent states is based on the equivalence, under linear optical transformations, between the probability of detecting a single photon at alternative output modes, and the relative intensity of corresponding coherent states. In particular, assuming $|\vec{\alpha}|=1$, the input vector $\ket{\vec{\alpha}}$ maps to 
\begin{equation}
  \ket{\psi(\vec{\alpha})} = \sum_k \alpha_k a_k^{\dag} \ket{\mathfrak{0}},
\end{equation}
where $\ket{\mathfrak{0}}$ denotes the vacuum state and $ a_k^{\dag}$ is the bosonic creation operator in the mode $k$.

The probability of measuring the photon at detector $j$ is then given by
\begin{equation}
  \mathbb{P}(j|\vec{\alpha}) = \left| \sum_k M_{jk} \alpha_k \right|^2.
  \label{eq:experiment.probabilities}
\end{equation}
By taking into account the statistical fluctuations introduced by a finite-sample frequentist estimate of the probability, this is analogous to the noisy intensity measurements described in \eqref{eq:intensities}.

For both reconstructions, the photons from the free space source were coupled into and out of the photonic chip via pre-packaged polarisation maintaining optical fibres and all output modes were simultaneously measured by an array of single photon avalanche photodiodes.
For all transfer matrices $\mat M$ of the same size, the PhaseLift data was collected for the same set of randomly chosen input vectors. The overall number of input vectors recorded is summarised in table \cref{tab:measurements1}.
Further technical details are reported in the \supplementary. 


\begin{table}[bt]
  \begin{tabular}{|l | r r r|}
  \hline
    Dimension $n$ & 2 & 3 & 5 \\
    \hline
    Gaussian & 20 & 30 & 40 \\
    RECR & 6 & 31 & 39 \\
    \hline
  \end{tabular}
  \caption{%
    \label{tab:measurements1}
    Total number of distinct input vectors used during the experiment.
  }
\end{table}


Since the reconstruction obtained from two-photon interference is also oblivious of the column phases of the matrix, to compare the two-photon and the PhaseLift reconstructions we report the Frobenius distance between the two matrices after optimise the row phases as well as the column phases: 
\begin{equation}\label{eq:overfitting}
  \min_{\vec{\mu},\vec{\nu}}\left\| \mat{M}^\sharp -  \vec{D} (\vec{\mu}) \mat{M}_{\text{2photon}}\vec{D} (\vec{\nu}) \right\|_2.
\end{equation}
For short, we will label $\mat{M}_{\text{2ph}}^{\mu \nu} = \vec{D} (\vec{\mu^*}) \mat{M}_{\text{2photon}}\vec{D} (\vec{\nu^*})$ the optimal two-photon reconstruction obtained using the corrective phase $\vec{\mu^*} $ and $\vec{\nu^*} $ that minimise \cref{eq:overfitting}.

In \cref{fig:experimental.phase_transition}(a) we show the results from the 5-dimensional PhaseLift reconstructions. The distance from the two-photon reference is reported as a function of the number of input vectors $m$ used for the PhaseLift reconstruction. 
The large set of input vectors used for the characterisation measurements allows us to average the reconstruction distance over multiple combinations of input vectors. 

From the plot we observe that the performance of the uniform and RECR ensembles are qualitatively similar but with a slightly better agreement with the reference shown by the uniform ensemble. From the reconstruction distance at $m=4n-4$, we obtain an indication of the noise level affecting the reconstruction system $\sigma_\text{Uniform}\sim0.025$ and $\sigma_\text{RECR}\sim0.035$.
The experiment also clarifies how, in a real case scenario, 
the improvement of the reconstruction
continues after reaching the suggested number of input vectors, although at a lower pace. Indeed, while the theoretical performance  guarantee  is   valid for both systematic and stochastic source of noise, the PhaseLift algorithm permits to use larger datasets to reduce the noise on the reconstructed matrix due to stochastic errors.  

Choosing to fix $m=6n$, in \cref{fig:experimental.overview}(b) we compare the PhaseLift and the reference reconstruction for all tested dimensions by means of the circuit fidelity \cite{Carolan2015-Universal}, here defined as:
\begin{equation}\label{eq:fidelity}
    \mathcal{F}\left(\mat{A},\mat{B}\right) = \frac{
        \tr\left( \left| \mat{A}^\dag . \mat{B} \right|^2 \right)}{n},
\end{equation}
where $n$ is the dimension of the square transfer matrices $\mat{A}$ and $\mat{B}$, and the absolute squaring operation of the matrix product is computed element-wise. 
When computed between two unitary matrices, such fidelity has a clear procedural meaning. It is the probability of projecting a photon prepared according to a column of $\mat{B}$ onto the corresponding column of $\mat{A}$, averaged over the $n$ columns. When comparing unitary matrices, the fidelity \cref{eq:fidelity} is upper-bounded by 1, however, the particular definition has inconsistent behaviour for matrices with non-normalised columns. Therefore, the data reported with dots in  \cref{fig:experimental.overview}(b) refer to the fidelity between the unitary approximation of $\mat{M}^\sharp$ and $\mat{M}_{\text{2ph}}^{\mu \nu}$ obtained by means of polar decomposition that provides us with the closest unitary matrix to a square matrix $\mat{A} $ as defined by unitarily invariant norms \cite{Fan1955-PolarInequality}. Error bars represent the standard deviation observed by choosing alternative sub-samples of $6n$ input vectors from the measured ones. The lack of more than 6 independent 2-dimensional RECR vectors forced us to only use a fixed set of input vectors for this configuration.
For 2-, 3-, and 5-dimensional matrices, the average fidelity of the three Haar random unitaries was 0.9997 (0.99999), 0.9985 (0.9993), 0.993 (0.989)   when using the uniform (RECR) ensemble.  
As a comparison, we report with crosses on the same graph the average fidelity obtained between the non-unitary original reconstructed matrices $\mat{M}^\sharp$ and $\mat{M}_{\text{2ph}}^{\mu \nu}$. Without applying the polar decomposition, fidelity above 1 is sometimes observed for lower dimensional matrices while the noise in the reconstruction strongly penalises the 5-dimensional cases. Interestingly, an average fidelity above 0.988 (0.98) is also observed if the columns of the 5-dimensional PhaseLift reconstructed matrices are normalised without imposing their orthogonality. 

%
%

\section*{Conclusions}%
\label{sec:conclusion}
In this work, we introduce a practical solution to the problem of characterising linear optical devices based on recent advances in phase retrieval and low-rank matrix recovery.
The PhaseLift reconstruction outlined in \cref{prot:characterization} can be used to reconstruct transfer matrices using only intensity measurements and classical coherent light as input, by modulating its amplitude according to complex vectors chosen at random from an appropriate ensemble.
Not only do we provide numerical and experimental evidence that the number of illumination settings required for this approach scales linearly with the number of modes of the device, but we also support these findings with rigorous theory: PhaseLift ensures stability with respect to additive noise corruptions. This theoretical support extends, in particular, to 
the RECR ensemble, which holds a great potential for applications in linear optical devices with tailored amplitude network designs.

Due to the additional experimental overhead associated with the calibration of a phase-stable amplitude distribution network, the PhaseLift reconstruction is particularly suited for integrated devices that can be re-programmed to implement a family of different transformations. 
We demonstrated the successful implementation of the PhaseLift characterisation protocol on a universally reconfigurable six waveguide device.
The results from this experiment show that, even without an ad hoc optimisation of the distribution network, the performance of the RECR ensemble for the PhaseLift reconstruction is close to the more conventional uniform distribution.  


\subsection*{Acknowledgements}
We are thankful to Jacques Carolan for his work on setting up the optical chip.
We acknowledge the experimental support from Patrick Yard in the collection of two-photon interference data. 
We acknowledge support from the Engineering and Physical Sciences Research Council (\mbox{EPSRC}) Hub in Quantum Computing and Simulation (EP/T001062/1), and the U.S. Army Research Office (ARO) grant W911NF-14-1-0133 and Germany’s Excellence Strategy - Cluster of Excellence \emph{Matter and Light for Quantum Computing} (ML4Q) EXC2004/1. Fellowship support from EPSRC is acknowledged by A.L. (EP/N003470/1).
The authors are grateful for those who provide support for the following software packages: NumPy~\cite{Walt_2011_Numpy}, Cvxpy~\cite{Diamond_2016_Cvxpy}, IPython~\cite{Perez_2007_Ipython}, matplotlib~\cite{Hunter_2007_Matplotlib}, Seaborn~\cite{Waskom_2017_Seaborn}, and Pandas~\cite{Mckinney_2010_Data}.

\subsection*{Author contribution}
The project was conceived and managed by A.L. and D.G.
The mathematical demonstrations of PhaseLift reconstruction properties with uniform and RECR ensembles were developed by D.S., R.K. and D.G.
The experiment was designed by C.S., N.Mar., D.S. and A.L.
The original code to perform the PhaseLift reconstruction was developed by D.S. who also produced the numerical simulations presented in the text.
Experimental data were collected by N.Mar. and A.M.
Two-photon reconstruction was done by N.Mar. and A.M. while the analysis of the PhaseLift datasets and their presentation involved N.Mar, A.M., D.S. and C.S. 
The photonic chip was fabricated by N.Mat. and T.H.
The manuscript was prepared by D.S., R.K., N.Mar., A.M. and C.S. 
All authors edited the paper and contributed to its final revision.

\bibliographystyle{IEEEtran}
\bibliography{references_we_use}

\onecolumngrid

\appendix
\newpage
 
\section{Experimental Setup}%
\label{sec:experimental_details}

\subsection{Photon source}
We used a Titanium:Sapphire laser (Coherent Chameleon) to generate to generate 140fs long pulses centred at $808$nm with a repetition rate of $80$MHz.
A half-wave plate and a polarising beamsplitter (PBS) are used to attenuate the power.
Next, a $\beta$-barium borate (BBO) crystal is used to perform second harmonic generation.
Dichroic mirrors remove the remaining 808nm light and a 0.5mm thick Bismuth Triborate BiB$_3$O$_6$ (BiBO) crystal is used to perform spontaneous parametric down-conversion (SPDC) from the up-converted 404nm pulse.
Down-converted photons are emitted in a cone at opening angle $\theta = 6^{\circ}$ and pass through a 3nm interference filter centered at 808nm.
Light is collected from opposite points on the SPDC emission cone and coupled into polarisation maintaining fibres (PMF).

When either pair is connected directly to the detectors, the ratio between the coincidence detection rate and the single detection rate is $\sim12\%$.
Taking into account the detector efficiency, this translates to a heralding efficiency of around 24$\%$.

\subsection{Integrated Circuit}
The silica-on-silicon integrated photonic chip was fabricated by the Nippon Telegraph and Telephone company (NTT) in Japan.
Flame hydrolysis deposition followed by photolithographic and reactive ion etching was used to fabricate germanium doped silica (SiO2-GeO2) waveguides with dimensions 3.5$\mu$m $\times$ 3.5$\mu$m with a silica cladding onto a silicon substrate.
Thin-film Tantalum Nitride (Ta$_2$N) thermo-optic heaters were then fabricated on top of the circuit with dimensions 1.5mm $\times$ 50$\mu$m.
The circuit is formed of a cascaded array of 30 directional couplers (each with a length of 500$\mu$m) and 30 phase shifters designed to perform a universally reconfigurable transfer matrix on six waveguide modes.

The coupling losses have been estimated as $\sim 9\%$ per facet and the directional couplers at $<2.3\%$.
The average loss fibre-to-fibre was measured to be $\sim 42\%$.
The device is actively cooled via a Peltier cooling unit.

Thermo-optic modulators are driven by electronic heater driver boards which can deliver up to 20V with 4.9mV resolution and current up to 100mA.
These are then interfaced with a computer to set all the heaters to implement a given transfer matrix.

\subsection{Photon Detectors}
The detection system uses 6 SPADs (Perkin Elmer SPCM-AQRH-14), each with efficiencies $50-60\%$, a dark count rate of $\sim 100$Hz, timing jitter of $\sim 350$ps and a dead time of 32ns.
A coincidence counting card time-tagging all simultaneous channels in a time window usually set to be around 2ns is used to register detection events.
For each channel it is possible to set a specific time delay that is used by the counting card to compensate for the discrepancy in the signals arrival time introduced in the experiment by optical fibres, detectors, electronics and coaxial cables.

The detector efficiencies were estimated as follows.
Light was injected into the top mode of the circuit and counts were collected for 100 Haar-random unitary configurations of the circuit.
The set of relative efficiencies that minimised the sum of the total variation distances of the measured distributions to their targets was then used.
Experimental counts are adjusted by these estimated relative efficiencies.

\section{Two-photon Reconstruction}
Since our goal is to test the PhaseLift characterisation technique, and not the performance of the optical chip, we compare the PhaseLift reconstructions against the reconstruction obtained with a different experimental technique.
These reference reconstructions are obtained with a variant of the methods described in \cite{laing2012-superstable,Tillmann2016-reconstruction}.
For each matrix,  $n\times n$  single photon data is recorded by routing the heralded single photons alternatively into each of the input mode $k$ of the transfer matrix and recording the coincidence between each detector and the heralding signal. These provide us with the information about the modulus squared, $\rho_{jk}^2=\abs{M_{jk}}^2$, of the matrix elements. In particular, for each input mode, these counts are corrected with the detector efficiency calibration and divided by their cumulative sum to produce a normalised column of the matrix. 2-photon interference data is then collected to determine the phases of these matrix elements. 

We estimate the phase of each component by following a similar approach to~\cite{laing2012-superstable} that is based on the measurement of several HOM-dips~\cite{Hong1987-HOMdip}.
To perform the HOM experiments,  MZIs of the amplitude distribution stage, marked red in \cref{fig:experimental.schematic}, are set to perform an identity transformation and both fibres carrying the photon pairs generated by the SPDC source are connected to input ports of the chip.
For each pairwise combination of the input modes of our matrices, we record the twofold coincidences among our detectors while changing the time delay of a photon relative to the other via a motorised translation stage. The only exception regards the input combination into modes 4 an 5 that is not available in our setup. 
For each pair of input modes, prior to the matrix characterisation, we record a reference signal while implementing a balanced beam splitter on the device. As detailed later, this provides us with information about the position in the motorised scan when the two photons are indistinguishable as well as optical properties of the wavepackets.  

In each two photon interference measurement, the coincidence counts as a function of delay are fit to the function:
\begin{equation}\label{eq:dip_fitting}
  f(\tau)=\left(1-c_1\exp[-((\tau - c_2)/c_3)^2]\ \mathrm{sinc}[(\tau - c_2)/c_6]\right)(c_4\tau + c_5 + c_7 \tau^2) 
\end{equation}
where $\tau$ is the time delay of the photon and $\{c_i\}$ fit parameters.
The first term approximates the temporal envelope of a Gaussian photon subject to a top-hat filter and the second term adjusts for the decoupling resulting from the movement of the translation stage.
The coefficient $c_1$ of this fit constitutes a bare visibility of the two photon interference later normalised by knowing the value of the reference dip.

The coefficients $c_2,\ c_3,\ c_6$ recorded from the reference dips are used as starting point to fit the signal obtained while characterising the matrices. When the estimated visibility of a signal is comparable with the noise level, we choose to constrain the parameters $c_2,\ c_3,\ c_6$ of the fit to those of the reference to prevent overfitting.
Initial values for the coefficients $c_1,\ c_4,\ c_5,\ c_7$ are originated by a few arithmetic combinations of the minimum, maximum and average   of the dataset together with the values at the extremes of the scanning interval of the translation stage. 
Before proceeding to the determination of the matrix elements, the visibilities obtained by fitting \cref{eq:dip_fitting} are divided by the reference visibility to account for the partial-distinguishability of the photons generated by our source that in the different measurements ranged between from 0.965 and 0.98.

Following~\cite{laing2012-superstable}, we first determine the absolute values of the phases of all matrix elements with the following algorithm.
\begin{enumerate}
\item We assume that all the elements of the first row and of the first column of the matrix are real and positive. 
\item By using all the input combinations of the form $(1,k)$ with $1<k\leq n$ and all the outputs combinations $(1,j)$ with $1<j\leq n$, we set the absolute values of the phase of the elements  $M_{jk}$ to be 
\begin{equation}
\left|\phi_{jk}\right|= \arccos\left( - \frac{|M_{11}M_{jk}|^2+|M_{1k}M_{j1}|^2}{2 |M_{11}M_{jk}M_{1k}M_{j1}|} V_{1k,1j} \right)\in \left[0,\pi \right],
\end{equation}
where $V_{1k,1j}$ is the visibility observed injecting photons into the ports 1 and $k$ and detecting photons at the output ports 1 and $j$.
\end{enumerate}

The sign of the phases is then attributed with the following routines:
\begin{enumerate}
\item We impose that the phase of the element $M_{22}$ is between 0 and $\pi$: $\phi_{22}=|\phi_{22}|$
\item For elements $M_{j2}$ with $2<j\leq n$, we attribute the following sign to the phase as determined by using the visibilities from the input combination $(1,2)$ with outputs combinations $(2,j)$. In particular, once defined 
\begin{equation}
    \eta_{j2}=  \arccos\left( - \frac{|M_{21}M_{j2}|^2+|M_{22}M_{j1}|^2}{2 |M_{21}M_{j2}M_{22}M_{j1}|} V_{12,1j} \right),
\end{equation}
we use $\eta_{j2}$ to set the sign of $\phi_{j2}$ as 
\begin{equation}
    \text{sign}\left[ \phi_{j2} \right] =
    \text{sign}\Big[\big|\eta_{j2}-\arccos\big(\cos(\phi_{22}+|\phi_{j2}|)\big)     \big|-
    \big|\eta_{j2}-\arccos\big(\cos(\phi_{22}-|\phi_{j2}|)\big) \big|
    \Big]
\end{equation}
\item We then attribute the sign to each element phase $\phi_{jk}$ with $1<j\leq n$ and $2<k\leq n$ by using  the visibilities from the input combination $(2,k)$ with outputs combinations $(1,j)$. In details, we set
\begin{equation}
    \eta_{jk}=  \arccos\left( - \frac{|M_{22}M_{jk}|^2+|M_{2k}M_{j2}|^2}{2 |M_{22}M_{jk}M_{2k}M_{j2}|} V_{2k,1j} \right),
\end{equation}
and determined $\phi_{jk}$ as
\begin{equation}
     \phi_{jk}   =|\phi_{jk}|\cdot \text{sign}\Big[\big|\eta_{jk}-\arccos\big(\cos(\phi_{j2}+|\phi_{jk}|)\big)     \big|-
    \big|\eta_{jk}-\arccos\big(\cos(\phi_{j2}-|\phi_{jk}|)\big) \big|
    \Big]
\end{equation}
\end{enumerate}

The final matrix reconstruction, $M_{\text{2photon}}$, is obtained as the result of an optimisation protocol based on all the available observed data. The starting point of the optimisation is the transfer matrix obtained with the algorithm described so far. The cost function we chose to optimise is
\begin{equation}
\textit{F}=\sqrt{ \sum_{x=(a,b,c,d)} \left|V_x-\frac{\left|M_{ac}M_{bd}+M_{ad}M_{bc}\right|^2}{\left|M_{ac}M_{bd}\right|^2+\left|M_{ad}M_{bc}\right|^2}\right|^2  w_x } + \lambda \sqrt{0.5 \left\| M^\dag.M - \mathbb{I} \right\|_2^2 + 0.5 \left\| M.M^\dag - \mathbb{I} \right\|_2^2 } ,
\end{equation}
where $x$ is an index that runs over all the experimentally measured visibilities $V_x$, with $(a,b)$ being the input modes and $(c,d)$ being the output modes. $w_x$ are weights proportional to the number of counts contributing to the estimation of the different HOM interference signals; higher count rates, strongly dependent on the $|M_{2,2}M_{j,k}|^2+|M_{2,k}M_{j,2}|^2$, lead to better resolved HOM fringes. $\mathbb{I}$ is the identity matrix of the appropriate size.
$\lambda$ is an empirical factor to weight the term of the cost function that enforces unitarity constraints against the part of the cost function that is based on experimental observation. We note that for the 5-dimensional matrices, $x$ runs over 90 experimental observations while for 2- and 3- dimensional matrices the same number reduces to 1 and 9, respectively.

During the optimisation of $\textit{F}$, the matrix elements are decomposed as $M_{jk}=\rho_{jk}\cdot\gamma_{j}\cdot\delta_{k}\cdot e^{i\phi_{jk}}$ where $\rho_{jk}$ result from experimental single photon measurements and $\gamma$, $\delta$, and $\phi$ are real valued variables. We performed 20 sequential optimisation steps alternatively optimising  the phase degrees of freedom $\phi_{jk}\in(-\pi,\pi)$ or the losses degrees of freedom $\gamma_{j}$ and $\delta_{k}$.
We note that some of the degrees of freedom of the matrix are redundant because of the symmetry of the cost function: the list includes the phases of the first row and first column elements and $\gamma_1$ .

\subsection{PhaseLift reconstruction}%
\label{sub:experimental_details.data}

\begin{table}
  \begin{tabular}{l | r r r}
    Dimension $n$ & 2 & 3 & 5 \\
    Gaussian & 20 & 30 & 40 \\
    RECR & 6 & 31 & 39 \\
  \end{tabular}
  \caption{%
    \label{tab:measurements}
    Total number of preparation vectors taken during experiment.
  }
\end{table}

As mentioned in the main text, we estimate the intensity measurements from single photon counting rates.
After correcting for detector efficiency, all counting rates are scaled by a constant such that the resulting intensities obey $ \sum_j I_j(\alpha^{(l)}) = 1$.
The number of different input vector we injected into each characterised transfer matrix is reported in \cref{tab:measurements}.
We provide a ready for use implementation of the PhaseLift convex program~\eqref{eq:PhaseLift} as well as related algorithms in the open source library \textsc{pypllon}~\cite{Suess_2017_Pypllon},

In an ideal experiment, $M^\sharp$ would be unitary and, therefore, every row would have unit norm.
However, due to loss in the characterised circuit as well as detector inefficiencies, the norm of each row is smaller than one.
Since we cannot distinguish the two sources of loss in our current experimental setup, we cannot characterise the absolute photon loss in the circuit, but only the relative losses of the rows.

Raw data as well as the analysis scripts are available at \url{https://github.com/dseuss/phaselift-paper}.

\section{Recovery guarantee for phase retrieval via PhaseLift}%
\label{sec:guarantee_phase_retrieval}

In this section, we provide the necessary background and convergence proofs for phase retrieval via the PhaseLift algorithm in a self-contained manner.
This approach is designed to recover arbitrary vectors $\vec x \in \mathbb{C}^n$ from noisy, phaseless measurements of the form
\begin{equation}
y_k = \left| \langle \vec\alpha_k, \vec x \rangle \right|^2 + \epsilon_k \label{eq:guarantee.measurements}
\end{equation}
via solving the convex optimisation
\begin{align}
  \underset{\mat{Z} \in \Hermitian_n}{\textrm{minimize}} & \quad \sum_{l=1}^m \left| \tr \left( (|\vec\alpha^{(l)} \rangle \langle \vec\alpha^{(l)} |) \mat{Z} \right) - y^{(l)} \right| \label{eq:guarantee.phaselift} \\
   \textrm{subject to} &\quad  \mat{Z} \geq 0. \nonumber
\end{align}
Here, $\Hermitian_n$ denotes the set of hermitian $n \times n$ matrices.
We focus on measurements $\vec\alpha_1,\ldots,\vec\alpha_m \in \mathbb{C}^n$ chosen independently from a distribution that obeys three conditions:
\begin{itemize}
\item \emph{Isotropy on $\mathbb{C}^n$:}
\begin{equation}
  \mathbb{E} \left[ | \langle \vec\alpha, \vec z \rangle |^2 \right] = C_I \| \vec z \|_{\ell_2}^2 \quad \forall \vec z \in \mathbb{C}^n. \label{eq:tight_frame}
\end{equation}

\item \emph{Sub-Isotropy on $\Hermitian_n$}:
\begin{align}
  \mathbb{E} \left[ \langle \vec\alpha | \mat Z | \vec\alpha \rangle^2 \right] \geq C_{SI} \| \mat Z \|_2^2 \quad \forall \mat Z \in \Hermitian_n \label{eq:sub_isotropy}
\end{align}

\item \emph{Sub-Gaussian tail behaviour:} For every normalized $\vec z \in \mathbb{C}^n$ ($\| \vec z \|_{\ell_2}=1$), $| \langle \vec\alpha, \vec z \rangle|$ is sub-Gaussian in the sense that its moments obey
\begin{equation}
  \mathbb{E} \left[ | \langle \vec\alpha, \vec z \rangle|^{2N} \right] \leq C_{SG} N! \quad N \in \mathbb{N}.
\label{eq:subexponential}
\end{equation}
\end{itemize}

\begin{proposition}%
  \label{prop:gauss+recr_requirements}
  The following measurement ensembles fulfil the three properties above:
  \begin{enumerate}

    \item \emph{Gaussian sampling scheme:} $\vec\alpha \in \mathbb{C}^n$ is chosen from the standard complex normal distribution $\mathcal{N}(0,\tfrac{1}{2}\mathbb{I})+ i \mathcal{N}(0,\tfrac{1}{2}\mathbb{I})$. In this case
  \begin{equation*}
    C_I = C_{SI} = C_{SG} = 1.
  \end{equation*}

  \item \emph{Uniform sampling scheme:} $\vec\alpha \in \mathbb{C}^n$ is a vector chosen uniformly from the complex unit sphere with radius $\sqrt{n}$. In this case
  \begin{equation*}
    C_I = 1, \; C_{SI} = \frac{n}{n+1}, \; C_{SG} = \prod_{k=1}^{N-1} \frac{n}{n+k} \leq 1.
  \end{equation*}

  \item (unnormalised) \emph{Randomly Erased Complex Rademacher (RECR) sampling scheme:} with respect to a fixed (arbitrary) basis, the coefficients of $\vec\alpha \in \mathbb{C}^n$ are chosen from the following distribution:
  \begin{equation}
  \alpha_i \sim
  \begin{cases}
  +1 & \textrm{with prob. } p/4 \\
  +\ii & \textrm{with prob. } p/4 \\
  0 & \textrm{with prob. } 1-p \\
  -\ii & \textrm{with prob. } p/4 \\
  -1 & \textrm{with prob. } p/4 \\
  \end{cases}
  \label{eq:definition_recr}
  \end{equation}
  The constants depend on the erasure probability $1-p \in [0,1]$:
  \begin{equation*}
  C_I = p,\; C_{SI} = p \min \left\{p,1-p \right\}, \; C_{SG} = \mathrm{e}^{\frac{3}{2}}.
  \end{equation*}
  \end{enumerate}
\end{proposition}

The proof techniques developed here do not apply to the normalized RECR scheme presented in the main text.
We comment on potential extension to this case in \cref{sec:normalized_recr}.\\

We now turn to the problem of proving recovery guarantees for PhaseLift with inputs sampled from distributions satisfying the conditions stated above.
The following statement is a substantial generalisation of existing results regarding phase retrieval from Gaussian and uniform measurements~\cite{candes_solving_2012,demanet_stable_2014}:

\begin{theorem}[Theorem 1.3 in~\cite{candes_solving_2012}]%
  \label{thm:phaselift_noisy}
  Suppose that $m = Cn$ vectors $\vec\alpha_1,\ldots,\vec\alpha_m \in \mathbb{C}^n$ have been chosen independently at random from an ensemble that obeys the three properties \eqref{eq:tight_frame}, \eqref{eq:sub_isotropy} and \eqref{eq:subexponential}.
  Then, with probability at least $1 - 3\mathrm{e}^{-\gamma m}$,  $m$ noisy measurements of the form~\eqref{eq:guarantee.measurements} suffice to reconstruct any $\vec{x} \in \BC^n$ via PhaseLift -- the convex optimisation problem~\eqref{eq:guarantee.phaselift}.
  This reconstruction is stable in the sense that the minimizer $\mat{X}^\sharp$ of \cref{eq:guarantee.phaselift} is guaranteed to obey
  \begin{equation}
    \left\| \vec{X}^\sharp - |\vec{x} \rangle \! \langle \vec{x}| \right\|_2 \leq \frac{C'  \| \epsilon \|_{\ell_1}}{m}. \label{eq:noisy_recovery_bound}
  \end{equation}
Here, $\| \cdot \|_2$ denotes the Hilbert-Schmitt norm $\| \mat{Z} \|_2^2 = \tr \left( \mat{Z} \mat{Z}^\dagger \right)$, while $C,C'$ and $\gamma$ represent constants of sufficient size.
\end{theorem}

The constants $C,C',\gamma$ implicitly depend on $C_I,C_{SI}, C_{SG}$ in \eqref{eq:tight_frame}, \eqref{eq:sub_isotropy}, \eqref{eq:subexponential} and can in principle be extracted from the proof.
No attempt has been made to optimize them.
Note that the demand on the number of measurements $m$ in \cref{thm:phaselift_noisy} scales linearly in the problem's dimension $n$.
This is optimal up to the constant multiplicative factor $C$.
Analytical bounds on this constant $C$ are usually too pessimistic to be practical and it is widely believed that
\(
  m = 4n - 4
\)
such measurements are actually sufficient, that is $C = 4 + o(n)$~\cite{heinosaari_quantum_2013}.
We refer to~\cite{MixonBlog} for further information about this topic.

We postpone the proof of this result to \cref{sec:main_proof} in order to derive an error bound for the signal vector from \cref{thm:phaselift_noisy}.
Recall that we obtain the recovered signal vector $\vec{x}^\sharp$ from the minimizer $\mat{X}^\sharp$ of \cref{eq:guarantee.phaselift} by a eigenvalue decomposition.
Let
\begin{equation}
  \mat{X}^\sharp = \sum_i \lambda_i \ket{\vec x_i}\bra{\vec x_i}
\end{equation}
with $\lambda_1 \ge \lambda_2 \ldots \ge \lambda_n$.
Then, we set
\begin{equation}
  \vec{x}^\sharp = \sqrt{\lambda_1} \vec x_1.
\end{equation}
In~\cite{candes_solving_2012} it was shown that \cref{eq:noisy_recovery_bound} implies
\begin{equation}
  \min_{0 \leq \phi \leq 2 \pi} \, \| \vec{x}^\sharp - \mathrm{e}^{\ii \phi} \vec{x} \|_{\ell_2}
  \leq C'' \frac{\| \epsilon \|_{\ell_1} }{m \| \vec{x} \|_{\ell_2}}
, \label{eq:vectorial_noisy_bound}
\end{equation}
where $C''$ again denotes an absolute constant.

\section{Recovery guarantee for \cref{prot:characterization}}%
\label{sec:guarantee}

Now, let us investigate the consequences of \cref{thm:phaselift_noisy} to our problem of characterising linear optical networks.
In \cref{eq:intensities_as_overlap}, we have already underlined the similarity between the intensity measurements in our setup on one hand and the assumed measurements~\eqref{eq:guarantee.measurements} for the phase retrieval on the other hand.
The input vector $\vec{\alpha}$ acts as a measurement vector and each row of the transfer matrix $\vec{M}_j$ plays the role of a complex signal vector $\vec x$ to be determined.
\Cref{thm:phaselift_noisy} now allows for recovering these row vectors individually by means of the PhaseLift algorithm~\eqref{eq:guarantee.phaselift}.
In order to meet the requirements for recovering the first row of $\mat{M}$ using \cref{thm:phaselift_noisy}, we have to measure the intensities in the first output mode $I_1(\vec\alpha^{(l)})$ for $m = Cn$ random coherent input states $\ket{\vec{\alpha}^{(1)}}, \ldots \ket{\vec{\alpha}^{(m)}}$ sampled from a distribution fulfilling the conditions \eqref{eq:tight_frame}, \eqref{eq:sub_isotropy} and \eqref{eq:subexponential}.
Said theorem then guarantees recovery of $\mat{M}_1 \in \BC^n$ -- the complex conjugate of the first row of $\vec M$ with high probability by means of PhaseLift \eqref{eq:guarantee.phaselift}.

Before we can move on to determine the remaining row vectors $\mat{M}_{j}$ ($2 \leq j \leq n$) of $\mat{M}$, it is important to point out that the recovery guarantee of \cref{thm:phaselift_noisy} is \emph{universal}: one instance of randomly chosen measurement vectors suffices to recover \emph{any} vector $\vec{x} \in \BC^n$.
In our particular setting, this universality assures that a single choice of random coherent states $\ket{\vec{\alpha}^{(1)}},\ldots,\ket{\vec{\alpha}^{(m)}}$ suffices to recover all row vectors $\mat{M}_j$ simultaneously.
Putting everything together yields the following protocol:
\begin{protocol}[\emph{Reconstruction of the transfer matrix $\mat{M}$}]%
  \label{prot:detailed_reconstruction}
  Let $\mat{M}$ be an arbitrary $n \times n$ transfer matrix as defined in~\eqref{eq:coherent_transfer_matrix}.
  In order to approximately recover it, sample $m = Cn$ random coherent input states $\ket{\vec{\alpha}^{(1)}},\ldots,\ket{\vec{\alpha}^{(m)}}$ from a distribution satisfying conditions \eqref{eq:tight_frame}--\eqref{eq:subexponential} and measure the $mn$ intensities
  \begin{equation*}
    y_j^{(l)} = \left| \sum_i M_{j,i} \, \alpha_i^{(l)} \right|^2 + \epsilon_j^{(l)} \quad \forall 1 \leq j \leq n, \quad 1 \leq l \leq m,
  \end{equation*}
  where $\epsilon_j^{(l)}$ denotes the additive noise at detector site $j$ when measuring the intensity resulting from input state  $\ket{\vec{\alpha}^{(l)}}$.
  For each $1 \leq j \leq n$, solve the semi-definite program
  \begin{align}
    \mat{Z}^\sharp_{j} = \underset{\mat{Z} \in \Hermitian_n}{\argmin}& \quad \sum_{l=1}^m \left| \tr \left( (|\vec{\alpha}^{l} \rangle \langle \vec{\alpha}^{(l)} | ) \mat Z \right) - y_j^{(l)} \right| \label{eq:tmat_recovery_program}\\
    \mathrm{s.t.} & \quad \mat{Z} \geq 0 \nonumber
  \end{align}
  and let $\mat{M}_j^\sharp$ be the eigenvector of $\mat{Z}^\sharp_{j}$ corresponding to its largest eigenvalue and rescaled to have length $\norm{\mat{M}_j^\sharp}_{\ell_2} = \sqrt{ \| \mat{Z}^\sharp_{j} \|_\infty}$.
  Then, the we estimate $\mat{M}$ by
  \begin{equation}
    \mat{M}^\sharp =
    \begin{pmatrix}
      \adj{\vec{M^\sharp}_1} \\ \vdots \\  \adj{\vec{M^\sharp}_n}
    \end{pmatrix}.
    \label{eq:transfermat_estimator}
  \end{equation}
\end{protocol}
Note that \cref{eq:transfermat_estimator} simply amounts to stacking the separately recovered row vectors $\vec M_j^\sharp$.
The additional complex conjugation by taking the adjoint is due to the definition of $\vec M_j$ below \cref{eq:intensities_as_overlap}.

Now, a simple extension of \cref{thm:phaselift_noisy} yields a similar performance guarantee for \cref{prot:detailed_reconstruction}.
In order to succinctly state this result, we introduce some additional notation.
Define the total noise at detector site $j$ (measured in $\ell_1$-norm) to be
\begin{equation*}
  \epsilon_j^{\mathrm{tot}}= \sum_{l=1}^m \abs{\epsilon_j^{(l)}}
\end{equation*}
and the overall noise strength:
\begin{equation}
  \epsilon^{\mathrm{tot}} = \sqrt{ \sum_{j=1}^n {\epsilon_j^{\mathrm{tot}}}^2}.
\end{equation}

\begin{corollary}[Performance guarantee for \cref{prot:detailed_reconstruction}]
  \label{cor:noisy_performance_guarantee}
  The reconstruction $\mat{M}^\sharp$ of any transfer matrix $\mat{M}$ by means of \cref{prot:detailed_reconstruction} satisfies
    \begin{equation}
      \min_{\vec{\mu}: \abs{\mu_j} = 1} \left\|  \mat{M}^\sharp -  \vec{D} (\vec{\mu}) \mat{M} \right\|_2
      \leq C' \frac{n \epsilon^\mathrm{tot}}{m \nu}.
      \label{eq:noisy_reconstruction_total_bound}
    \end{equation}
  with high probability (i.e.\ with probability at least $1 - \CO \left( \mathrm{e}^{-\gamma m}\right)$).
  Here, $C'$ is a constant of sufficient size and
  \begin{equation}
    \nu = \min_{1 \leq j \leq n} \| \mat{M}_j \|_{\ell_2}.
    \label{eq:definition_normconst}
  \end{equation}
\end{corollary}
Recall that $\mat{D}(\vec\mu) = \mathrm{diag}(\mu_1, \ldots, \mu_n)$ with $\abs{\mu_j} = 1$ are the row-phases of $\mat{M}$ unrecoverable from the measurements~\eqref{eq:intensities}.
Note that $\nu = 1$ for unitary transfer matrices and $\nu < 1$ if there is loss.

\begin{proof}
For any fixed row vector $\vec{M}_j$,
\begin{equation}
  \min_{0 \leq \phi \leq 2 \pi}\left\| \mat{M}_j^\sharp - \mathrm{e}^{i \phi} \mat{M}_j \right\|_{\ell_2} \leq C' n \min
  \left\{
  \| \mat{M}_j \|_{\ell_2}, \frac{ \epsilon_j^{\mathrm{tot}}}{m \| \mat{M}_j \|_{\ell_2}}
  \right\}.
  \label{eq:noisy_reconstruction_vectorial_bound}
\end{equation}
follows directly from \cref{thm:phaselift_noisy} and \eqref{eq:vectorial_noisy_bound}, respectively. Universality (one choice of phaseless measurements allows for reconstructing any vector) moreover allows for applying this reconstruction guarantee to all $n$ row vectors $\vec{M}_j$ simultaneously.

The total noise bound~\eqref{eq:noisy_reconstruction_total_bound} follows from the entry-wise definition of the Frobenius norm:
  \begin{align*}
    \min_{\vec{\mu}}\left\|  \mat{M}^\sharp -  \vec{D} (\vec{\mu}) \mat{M} \right\|_2 ^2
    &= \min_{0 \leq \phi_1,\ldots,\phi_n \leq 2 \pi}
    \sum_{j=1}^n \left\| \mat{M}_j^\sharp - \mathrm{e}^{i \phi_j} \mat{M}_j \right\|_{\ell_2}^2 \\
    &= \sum_{j=1}^n \min_{0 \leq \phi_j \leq 2 \pi} \left\| \mat{M}_j^\sharp - \mathrm{e}^{i \phi_j} \mat{M}_j \right\|_{\ell_2}^2 \\
    & \leq (C')^2 n^2 \sum_{j=1}^n  \min \left\{ \| \mat{M}_j \|_{\ell_2}^2, \frac{ \eta_{(j)}^2}{m^2 \|\mat{M}_j \|_{\ell_2}^2} \right\} \\
    &\leq \left(C' n\right)^2 \sum_{j=1}^n \frac{\eta_{j}^2}{m^2 \|\mat{M}_j \|_{\ell_2}^2} \\
    & \leq \frac{\left(C'n \right)^2}{m^2 \nu} \sum_{j=1}^n \eta_{j}^2 \\
    &= \left( C' n \frac{\eta^{\mathrm{tot}}}{m \nu} \right)^2,
  \end{align*}
 Here, we have used \eqref{eq:vectorial_noisy_bound} for each summand in the third line.
  Taking the square root then yields the desired expression. 
\end{proof}
Note that this formulation allows to treat the different output modes and their detector noise levels individually.
In particular, we do not require a universal type of noise for all detectors, but allow for taking into account detector dependent noise of different quality (i.e.\ varying noise levels).


\section{Proof of \cref{thm:phaselift_noisy}}
\label{sec:main_proof}

Our analysis is inspired by Ref.~\cite{dirksen_gap_2015} (who derived strong results for sparse vector recovery using similar assumptions) and Ref.~\cite{kabanava_stable_2016} in the non-commutative setting. Moreover, Krahmer and Liu considered a real-valued version of the problem addressed here, see Ref.~\cite{krahmer_phase_2017}.

\subsection{Mathematical preliminaries}

Our analysis is based on two strong results about random matrix theory. First, the assumption of subgaussian tails \eqref{eq:subexponential} implies strong bounds on the operator norm of matrices of the form $\sum_{k=1}^m \ket{\vec{\alpha_k}}\bra{\vec\alpha_k}$:

\begin{theorem}[Variant of Theorem 5.35 in \cite{Vershynin_2010_Introduction}] \label{thm:bernstein}
Suppose that $\vec\alpha_1,\ldots,\vec\alpha_m$ are independent instances of a subgaussian random vector obeying \eqref{eq:subexponential} with constant $C_{SG}$.
Set
\begin{equation}
  \tilde{\mat H} = \frac{1}{m} \sum_{k=1}^m \left( a_k |\vec\alpha_k \rangle \! \langle \vec\alpha_k| - \mathbb{E} \left[ a_k |\vec\alpha_k \rangle \! \langle \vec\alpha_k| \right] \right),
  \label{eq:Htilde}
\end{equation}
where $a_k \in \mathbb{C}$ and $\abs{a_k} \leq 1$.
Then,
\begin{align*}
\mathrm{Pr} \left[ \| \tilde{\mat H} \|_\infty \geq t \right]
\leq
\begin{cases}
2 \exp \left( 2 \ln (3) n  - \frac{mt^2}{8 C_{SG}} \right) & 0 \leq t \leq 2C_{SG}, \\
2 \exp \left( 2 \ln (3) n - \frac{m}{2} (t- C_{SG} )  \right) & t \geq 2 C_{SG}.
\end{cases}
\end{align*}
\end{theorem}

The second result is a generalisation of ``Gordon's escape through a mesh''-Theorem \cite{gordon_milman_1988} (a random subspace avoids a subset provided the subset is small in some sense) that is due to Mendelson \cite{mendelson_learning_2015,koltchinskii_bounding_2015}, see also see also \cite{tropp_convex_2015}.

\begin{theorem}[Mendelson's small ball method] \label{thm:mendelson}
  Suppose that the measurement operator $\mathcal{A}:\Hermitian_n \to \mathbb{R}^m$ contains $m$ independent copies $\mat A_k$ of a random matrix $\mat A \in \Hermitian_n$, that is
  \begin{equation}
    \label{eq:measurement_operator_definition}
    \mathcal{A}(\mat Z) = \sum_{k=1}^m \tr (\mat A_k \mat Z) \, \vec e_k,
  \end{equation}
  and let $D \subset \Hermitian_n$.
  For $\xi >0$ define
  \begin{align}
    Q_\xi (D, \mat A) =& \inf_{\mat Z \in D}\mathrm{Pr} \left[ | \tr (\mat A_k \mat Z) | \geq \xi \right] \quad &\textrm{(marginal tail funtion)}, \label{eq:marginal_tail_function}\\
    W_m (D, \mat A) =& 2 \mathbb{E} \left[ \sup_{\mat Z \in D} \tr \left( \mat Z \mat H \right) \right] \quad &\textrm{(mean empirical width)},
  \end{align}
  where
  \begin{equation}
    \mat H= \frac{1}{\sqrt{m}} \sum_{k=1}^m \eta_k \mat A_k
  \end{equation}
  and the $\eta_1,\ldots,\eta_m$ are independent Rademacher random variables.
  Then for any $\xi >0$ and $t >0$
  \begin{equation}
    \frac{1}{\sqrt{m}}\inf_{\mat Z \in D} \| \mathcal{A}(\mat Z) \|_{\ell_1} \geq \xi \sqrt{m} Q_{2\xi}(D, \mat A) -  W_m (D, \mat A)-\xi t \label{eq:mendelson}
  \end{equation}
  with probability at least $1-\mathrm{e}^{-2t^2}$.
\end{theorem}

Note that the measurement operator introduced in \cref{eq:measurement_operator_definition} is a shorthand notation for the linear measurements $y_k = \tr \mat A_k \mat Z$ with $k=1,\ldots,m$.
It maps the signal matrix $\mat Z$ to the vector of (noiseless) measurement outcomes $\sum_k y_k \vec e_k$.

\subsection{Convex geometry}

This section summarizes several results presented in Ref.~\cite{kabanava_stable_2016} and adapts them to the task at hand: phase retrieval.
Compared to~\cite{kabanava_stable_2016} the analysis presented here is somewhat more direct and exploits the positive semidefinite constraint in a different way.

\begin{proposition} \label{prop:nsp_implication}
  Let $\Sphere^{n^2-1}=\left\{ \mat Z \in \Hermitian_n: \| \mat Z \|_2=1 \right\}$ be the (Frobenius norm) unit sphere in $\Hermitian_n$ and $\mathcal{B}_1 = \mathrm{conv} \left\{ \pm |\vec x \rangle \! \langle \vec x| \colon \vec x \in \Sphere^{n-1} \right\}$ denote the trace-norm ball.
  Define
  \begin{equation}
    D := \Sphere^{d^2-1} \cap 3 \mathcal{B}_1, \label{eq:D}
  \end{equation}
  and let $\mathcal{A}(\mat Z) = \sum_{k=1}^m \tr (\mat A_k \mat Z ) \, \vec e_k$ be a measurement operator that obeys
  \begin{align}
      \frac{ \tau}{m} \| \mathcal{A}(\mat Z) \|_{\ell_1} \geq& \norm{\mat Z}_2 \quad \forall \mat Z \in D \label{eq:nsp}\\
      \| \frac{1}{\nu m}\sum_{k=1}^m \mat A_k -  \mathbb{I} \|_\infty \leq& \frac{1}{6}\label{eq:approx_povm}
  \end{align}
  for some $\tau,\nu >0$.
  Then, the following relation holds for any $\mat Z \geq 0$ and any $|\vec{x} \rangle \! \langle \vec{x}|$:
  \begin{equation}
    \| \mat Z - |\vec{x} \rangle \! \langle \vec{x}| \|_2 \leq \frac{1}{m} \max \left\{ \tau, \frac{6}{\nu} \right\}  \| \mathcal{A}(\mat Z-|\vec{x} \rangle \! \langle \vec{x}|) \|_{\ell_1}. \label{eq:rec_guarantee}
  \end{equation}
\end{proposition}

\begin{proof}
In the proof we will frequently use the decomposition $\mat Z = \mat Z_1+\mat Z_c$ for $\mat Z$ with eigenvalue decomposition $\mat Z = \sum_{k=1}^n \lambda_k |\vec z^{(k)} \rangle \! \langle \vec z^{(k)}|$.
Then, $\mat Z_1 = \lambda_1 |\vec z^{(1)} \rangle \! \langle \vec z^{(1)}|$ is the leading rank-one component and $\mat Z_c = \mat Z-\mat Z_1$ is the ``tail''.
Note that, in particular, $\mat Z = \mat Z_1$ if and only if $\mat Z$ has unit rank.
Fix $\mat Z \geq 0$ and $|\vec{x} \rangle \! \langle \vec{x}|$.
\Cref{eq:rec_guarantee} is invariant under re-scaling, so we may w.l.o.g.\ assume $\| \mat Z-|\vec{x} \rangle \! \langle \vec{x}|\|_2=1$.
We treat the following two cases separately:
\begin{align}
I.) \quad& \| (\mat Z-|\vec{x} \rangle \! \langle \vec{x}|)_1 \|_1 \geq \frac{1}{2} \| (\mat Z-|\vec{x} \rangle \! \langle \vec{x}|)_c \|_1, \label{eq:nsp_case1} \\
II.) \quad & \| (\mat Z-|\vec{x} \rangle \! \langle \vec{x}|)_1 \|_1 < \frac{1}{2} \| (\mat Z-|\vec{x} \rangle \! \langle \vec{x}|)_c \|_1. \label{eq:nsp_case2}
\end{align}
Note that I.) implies
\begin{align*}
\| \mat Z-|\vec{x} \rangle \! \langle \vec{x}| \|_1 \leq &\| (\mat Z-|\vec{x} \rangle \! \langle \vec{x}|)_1 \|_1 + \| (\mat Z-|\vec{x} \rangle \! \langle \vec{x}|)_c \|_1 \leq 3 \| (\mat Z-|\vec{x} \rangle \! \langle \vec{x}|)_1 \|_1 \\
 = & 3 \| (\mat Z-|\vec{x} \rangle \! \langle \vec{x}|)_1 \|_2 \leq 3 \| \mat Z- |\vec{x} \rangle \! \langle \vec{x}| \|_2 = 3
\end{align*}
which in turn implies that $\mat Z-| \vec{x} \rangle \! \langle \vec{x}|$ is contained in $3 \mathcal{B}_1$.
Thus, \eqref{eq:nsp} is applicable and yields
\begin{equation*}
\| \mat Z - |\vec{x} \rangle \! \langle \vec{x}| \|_2 \leq  \frac{\tau}{m} \| \mathcal{A}(\mat Z-|\vec{x} \rangle \! \langle \vec{x}|) \|_{\ell_1}
\end{equation*}
which establishes \cref{eq:rec_guarantee} for case I in \eqref{eq:nsp_case1}.

For the second case, we use a consequence of von Neumann's trace inequality, see e.g. \cite[Theorem~7.4.9.1]{horn_topics_1991}: Let $\mat A, \mat B$ be matrices with singular values $\sigma_k (\mat A),\sigma_k (\mat B)$ arranged in non-increasing order.
Then
\begin{equation*}
  \| \mat A - \mat B \|_1 \geq \sum_{k=1}^d | \sigma_k (\mat A) - \sigma_k (\mat B)|
\end{equation*}
This relation implies
\begin{align*}
  \| \mat Z \|_1 =& \| |\vec{x} \rangle \! \langle \vec{x}| - (|\vec{x} \rangle \! \langle \vec{x}|-\mat Z) \|_1
  \geq \sum_{k=1}^d \left| \sigma_k (|\vec x \rangle \! \langle \vec x|) - \sigma_k (|\vec x \rangle \! |\langle \vec x|- \mat Z ) \right| \\
  \geq & \sigma_1 (|\vec x \rangle \langle \vec x|) - \sigma_1 \left( |\vec x \rangle \! \langle \vec x| - \mat Z \right)+ \sum_{k=2}^d \sigma_k \left( |\vec x \rangle \! \langle \vec x| - \mat Z\right) \\
  =&  \| |\vec x \rangle \! \langle \vec x| \|_1  - \| (|\vec x \rangle \! \langle \vec x| - \mat Z)_1 \|_1 + \|(|\vec x \rangle \! \langle \vec x| -\mat Z)_c \|_1 \\
  >& \| |\vec x \rangle \! \langle \vec x| \|_1 + \frac{1}{2} \| (|\vec x \rangle \! \langle \vec x|-\mat Z)_c \|_1,
\end{align*}
where the last inequality follows from \eqref{eq:nsp_case2}. Consequently,
\begin{align}
  \| |\vec x \rangle \! \langle \vec x| - \mat Z \|_1
  =& \| (|\vec x \rangle \! \langle \vec x| - \mat Z)_1 \|_1 + \| (|\vec x \rangle \! \langle \vec x|-\mat Z)_c \|_1
  \leq \frac{3}{2} \| (|\vec x \rangle \! \langle \vec x|- \mat Z )_c \|_1 \nonumber \\
  < & 3 \left( \| \mat Z \|_1 - \| |\vec x \rangle \! \langle \vec x| \|_1 \right). \label{eq:nsp_aux2}
\end{align}
Now, positive semidefiniteness of both $\mat Z$ and $\ket{\vec x}\bra{\vec x}$ together with assumption~\eqref{eq:approx_povm} implies
\begin{align*}
  \| \mat Z \|_1 - \| |\vec{x} \rangle \! \langle \vec{x}| \|_1
  =& \tr (\mat Z-|\vec{x} \rangle \! \langle \vec{x}|) =  \tr \left( \mathbb{I} \left( \mat Z-| \vec{x} \rangle \! \langle x|\right) \right) \\
  =&  \tr \left( \left( \mathbb{I} - \frac{1}{\nu m} \sum_{k=1}^m A_k \right) \mat Z-|\vec{x} \rangle \! \langle \vec{x}| \right) + \frac{1}{\nu m} \sum_{k=1}^m \tr \left( A_k (\mat Z-|\vec{x} \rangle \! \langle \vec{x}|) \right) \\
  \leq &  \left\|\mathbb{I}- \frac{1}{ \nu m} \sum_{k=1}^m A_k \right\|_\infty \| \mat Z-|\vec{x} \rangle \! \langle \vec{x}| \|_1 + \frac{1}{\nu m} \| \mathcal{A}(|\vec{x} \rangle \! \langle \vec{x}|-\mat Z) \|_{\ell_1} \\
  \leq &  \frac{1}{6 } \| \mat Z-|\vec{x} \rangle \! \langle \vec{x}| \|_1 + \frac{1}{\nu m} \| \mathcal{A}(|\vec{x} \rangle \! \langle \vec{x}|-\mat Z) \|_{\ell_1}.
\end{align*}
Inserting this into \eqref{eq:nsp_aux2} yields
\begin{align*}
\| |\vec x \rangle \! \langle \vec x| - \mat Z \|_1 < \frac{1}{2} \| |\vec{x} \rangle \! \langle \vec{x}|-\mat Z \|_1 +  \frac{3}{\nu m} \| \mathcal{A}(|\vec{x} \rangle \! \langle \vec{x}|-\mat Z) \|_{\ell_1}
\end{align*}
which implies the claim for case II in \eqref{eq:nsp_case2}.
\end{proof}

\begin{proposition} \label{prop:RECR_nsp}
  Under the assumptions of \cref{thm:phaselift_noisy}, the measurement operator
  \begin{equation}
    \label{eq:measurement_operator_rank1}
    \mathcal{A}(\mat Z) = \sum_k \tr \left(\ket{\vec a_k}\bra{\vec a_k} \mat Z\right) \vec e_k
  \end{equation}
  obeys both condition~\eqref{eq:nsp} and \eqref{eq:approx_povm} with probability at least $1- 3\mathrm{e}^{-\gamma m}$, provided that $C >1$ is sufficiently large.
\end{proposition}

We postpone the proof of this statement to \cref{sec:proof_measurement_operator_is_good} and directly derive \cref{thm:phaselift_noisy} -- which constitutes the main theoretical achievement of this work -- from this statement.

\begin{proof}[Proof of \cref{thm:phaselift_noisy}]
\Cref{prop:RECR_nsp} implies that a measurement operator~\eqref{eq:measurement_operator_rank1} containing $m \geq C n$ measurements sampled from a distribution satisfying \eqref{eq:tight_frame}, \eqref{eq:sub_isotropy} and \eqref{eq:subexponential} meets the requirements of \cref{prop:nsp_implication} with probability at least $1-3 \mathrm{e}^{-\gamma m}$.
Conditioned on this event, we have
\begin{equation}
\| \mat Z - |\vec{x} \rangle \! \langle \vec{x}| \|_2 \leq \frac{C'}{2m}  \| \mathcal{A}(\mat Z - |\vec{x} \rangle \! \langle \vec{x}|) \|_{\ell_1} \quad \forall \mat Z \geq 0,\; \forall \vec{x} \in \mathbb{C}^n,
\label{eq:nsp_implication2}
\end{equation}
where $C' = 2 \max \left\{\tau, 6/\nu \right\}$.
Now, suppose that we want to reconstruct a particular $\vec{x}$ from noisy measurements of the form $\vec{y} = \mathcal{A} (|\vec{x} \rangle \! \langle \vec{x}|) + \vec{\epsilon}$. Then Eq.~\eqref{eq:nsp_implication2} implies
\begin{align*}
\| \mat Z - |\vec{x} \rangle \! \langle \vec{x}| \|_2 \leq \frac{C'}{2m} \| \mathcal{A}(\mat Z) - \vec{y} + \vec{\epsilon} \|_{\ell_1}
\leq \frac{C'}{2m} \left( \| \vec{\epsilon} \|_{\ell_1} + \| \mathcal{A}(\mat Z) - \vec{y} \|_{\ell_1} \right)\quad \forall \mat Z \geq 0.
\end{align*}
PhaseLift -- the convex optimisation problem \eqref{eq:PhaseLift} -- minimizes the right hand side of this bound over all $\mat Z \geq 0$. Since $\mat Z = |\vec{x} \rangle \! \langle \vec{x}|$ is a feasible point of this optimisation, we can conclude that the minimizer $\mat Z^\sharp$ obeys
\begin{equation*}
\| \mathcal{A}(\mat Z^\sharp) - \vec{y} \|_{\ell_1} \leq \| \mathcal{A}(|\vec{x} \rangle \! \langle \vec{x}|)-\vec{y} \|_{\ell_1} = \| \vec{\epsilon} \|_{\ell_1}
\end{equation*}
which yields the bound presented in \eqref{eq:noisy_recovery_bound}.
\end{proof}

\section{Proof of \cref{prop:RECR_nsp}}
\label{sec:proof_measurement_operator_is_good}

\begin{lemma}[Bound on the marginal tail function]
  Let $D$ be the set introduced in \eqref{eq:D} and let $\mat A =|\vec a \rangle \! \langle \vec a| $, where $\vec a$ satisfies \eqref{eq:sub_isotropy} and \eqref{eq:subexponential}.
  Then, the marginal tail function \eqref{eq:marginal_tail_function} obeys
  \begin{equation*}
    Q_\xi (D, \mat A) \geq  C_Q \left( 1-  \frac{\xi^2}{C_{SI}}\right)^2  \quad \forall 0 \leq \xi \leq \sqrt{C_{SI}},
  \end{equation*}
  where $C_Q>0$ is a sufficiently small constant.
\end{lemma}

\begin{proof}
Fix $\mat Z \in D$, then $\| \mat Z \|_2 =1$ by definition of $D$.
Note that sub-isotropy \eqref{eq:sub_isotropy} and the Paley-Zygmund inequality imply for any $\xi \in [0,1]$
\begin{align*}
  \mathrm{Pr} \left[ | \langle \vec a| \mat Z |\vec a \rangle| \geq \xi \right]
  \geq & \mathrm{Pr} \left[ \langle \vec a| \mat Z |\vec a \rangle^2 \geq \frac{\xi^2}{C_{SI}} \mathbb{E} \left[ \langle \vec a|\mat Z|\vec a \rangle^2 \right] \right]
  \geq \left(1-\frac{\xi^2}{C_{SI}}\right)^2 \frac{\mathbb{E} \left[ \langle \vec a |\mat Z |\vec a \rangle^2 \right]^2}{\mathbb{E} \left[ \langle \vec a| \mat Z |\vec a \rangle^4 \right]}.
\end{align*}
Sub-isotropy ensures that the numerator is lower bounded by $C_{SI}^2 \| \mat Z \|_2^4 = C_{SI}^2$.
In order to derive an upper bound on the denominator, we use the constraint $\| \mat Z \|_1 \leq 3$ for any $\mat Z \in D$ together with the subgaussian tail behavior \eqref{eq:subexponential} of $\vec a$.
Insert an eigenvalue decomposition $\mat Z = \sum_{i=1}^n \lambda_i |\vec z^{(i)} \rangle \! \langle \vec z^{(i)}|$ (with $\lambda_i \in \mathbb{R}$ and $\vec z^{(i)} \in \Sphere^{n-1}$) and note
\begin{align}
  \mathbb{E} \left[ \langle \vec a| \mat Z |\vec a \rangle^4 \right]
  \leq & \sum_{i_1,i_2,i_3,i_4=1}^n | \lambda_{i_1} \lambda_{i_2} \lambda_{i_3} \lambda_{i_4} | \mathbb{E} \left[ \prod_{k=1}^4 | \langle \vec a, \vec z^{(i_k)} \rangle|^2 \right]. \label{eq:Q_aux1}
\end{align}
Now fix $\vec z^{(i_1)},\ldots,\vec z^{(i_4)}$ and use a combination of the AM-GM inequality and the fundamental relation between $\ell_p$-norms ($\| \vec v \|_{\ell_1} \leq k^{1-\frac{1}{k}} \| \vec v \|_{\ell_k}$ for $v \in \mathbb{R}^k$) to conclude
\begin{align*}
  \mathbb{E} \left[ \prod_{k=1}^4 | \langle \vec a,\vec z^{(i_k)}\rangle |^2 \right]
  \leq \frac{1}{4} \sum_{k=1}^4 \mathbb{E} \left[ | \langle \vec a, \vec z^{(i_k)} \rangle|^8 \right]
  \leq C_{SG} 4!,
\end{align*}
where the last inequality follows from condition \eqref{eq:subexponential}.
Consequently,
\begin{align*}
  \mathbb{E} \left[ \langle \vec a| \mat Z | \vec a\rangle^4 \right]
  \leq C_{SG} 4! \sum_{i_1,i_2,i_3,i_4} | \lambda_{i_1} \lambda_{i_2} \lambda_{i_3} \lambda_{i_4} |
  = 24 C_{SG} \| \mat Z \|_1^4 \leq 24*3^4 C_{SG},
\end{align*}
because $\mat Z \in D$ implies $\| \mat Z \|_1 \leq 3$.
In summary,
\begin{align*}
  \mathrm{Pr} \left[ | \langle \vec a| \mat Z |\vec a \rangle| \geq \xi \right]
  \geq \left(1-\frac{\xi^2}{C_{SI}}\right)^2 \frac{\mathbb{E} \left[ \langle \vec a| \mat Z |\vec a \rangle^2 \right]^2}{\mathbb{E} \left[ \langle \vec a| \mat Z |\vec a \rangle^4 \right]}
  \geq \left(1-\frac{\xi^2}{C_{SI}}\right)^2 \frac{C_{SI}^2}{1944C_{SG}}
\end{align*}
and the bound on $Q_\xi (D,\mat A)$ with $C_Q = \frac{C_{SI}^2}{1944 C_{SG}}$ follows from the fact that this lower bound holds for any $\mat Z \in D$.
\end{proof}

\begin{lemma}[Bound on the mean empirical width]
Let $D$ be the set introduced in \eqref{eq:D} and let $\mat H = \frac{1}{\sqrt{m}} \sum_{k=1}^m \eta_k | \vec\alpha_k \rangle \! \langle \vec\alpha_k|$, where each $\vec\alpha_k$ is subexponential in the sense of \eqref{eq:subexponential} and $m \geq \frac{2 \ln (3)}{C_{SG}} n$.
Then there exists a constant $C_W >0$ such that
\begin{equation*}
W_m (D,\mat A) \leq C_W \sqrt{n},
\end{equation*}
\end{lemma}

\begin{proof}

Note that by construction $D \subset 3 \mathcal{B}_1$ and consequently
\begin{align}
  W_m (D, \mat A) = 2 \mathbb{E} \left[ \sup_{\mat Z \in D} \tr (\mat Z \mat H) \right] \leq 6 \mathbb{E} \left[ \sup_{\mat Z \in \mathcal{B}_1} \tr (\mat Z \mat H) \right] = 6 \mathbb{E} \left[ \| \mat H \|_\infty  \right], \label{eq:Wm_hoelder}
\end{align}
where the last equality follows from the duality of trace and operator norm. Now note that $\tilde{\mat H} = \sqrt{m} \mat H$ is of the form \eqref{eq:Htilde}, where each $a_k$ is an independent Rademacher random variable.
\cref{thm:bernstein} thus implies
\begin{align}
  \mathrm{Pr} \left[\| \mat H \|_\infty \geq t \right]
  \leq
  \begin{cases}
   2 \times 9^n \exp \left( - \frac{t^2}{8 C_{SG}} \right) & t \leq 2C_{SG} \sqrt{m}, \\
  2 \times 9^n \exp \left( - \frac{\sqrt{m}}{2} \left( t - C_{SG} \sqrt{m} \right) \right) & t \geq 2 C_{SG} \sqrt{m}
  \end{cases}
  \label{eq:Wm_tails}
\end{align}
and we can bound $\mathbb{E} \left[ \| \mat H \|_\infty \right]$ by using the absolute moment formula,
see e.g.\ \cite[Propostion~7.1]{Foucart_2013_Mathematical}, and bounding the effect of the tails via \eqref{eq:Wm_tails}.
To this end, we split the real line into three intervals $[0, c \sqrt{n}], [c\sqrt{n}, 2 C_{SG} \sqrt{m}], [2 C_{SG} \sqrt{m},\infty[$, where $c$ is a constant that we fix later:
\begin{align*}
  \mathbb{E} \left[ \|\mat H\|_\infty \right] =& \int_0^\infty \mathrm{Pr} \left[ \|\mat H\|_\infty \geq t \right] \mathrm{d}t \\
  \leq & \int_0^{c \sqrt{n}} 1 \mathrm{d}t + 2 \times 9^n \left( \int_{c \sqrt{n}}^{2 C_{SG} \sqrt{m}} 2 \exp \left( - \frac{t^2}{8 C_{SG}} \right) \mathrm{d}t
  +  \mathrm{e}^{\frac{m C_{SG}}{2}} \int_{2 C_{SG} \sqrt{m}}^\infty \exp\left( - \frac{\sqrt{m}t}{2}  \right) \mathrm{d} t \right)\\
  \leq & c \sqrt{n} + 2 \times 9^n \left( \int_{c \sqrt{n}}^{2 C_{SG} \sqrt{m}}  \exp \left( - \frac{t^2}{8 C_{SG}} \right) \mathrm{d}t
  + \frac{2}{\sqrt{m}} \mathrm{e}^{-\frac{C_{SG} m}{2}}\right).
\end{align*}
For the remaining Gauss integral, we use $\frac{t}{c \sqrt{n}} \geq 1\; \forall t \geq c\sqrt{n}$ to conclude
\begin{align*}
  \int_{c \sqrt{n}}^{2 C_{SG} \sqrt{m}}  \exp \left( - \frac{t^2}{8 C_{SG}} \right) \mathrm{d}t
  \leq  \int_{c \sqrt{n}}^\infty \frac{t}{c \sqrt{n}}  \exp \left( - \frac{t^2}{8 C_{SG}} \right) \mathrm{d} t
  = \frac{8 C_{SG}}{c \sqrt{n}} \exp \left( - \frac{c^2 n}{8 C_{SG}} \right).
\end{align*}
Now, fixing $c = 4 \sqrt{\ln (3)C_{SG}}$ assures $\exp \left( -\frac{c^2 n}{8 C_{SG}}\right) = 9^{-n}$ and consequently
\begin{align*}
  \mathbb{E} \left[ \| \mat H \|_\infty \right]
  \leq & 4  \sqrt{ \ln (3) C_{SG} n} + \frac{4 \sqrt{C_{SG}}}{\sqrt{ \ln (3) n}} + \frac{4}{\sqrt{m}} \mathrm{e}^{2 \ln (3) n - C_{SG} m} \\
  \leq & 4\sqrt{C_{SG}} \left( \sqrt{ \ln (3) n} + \frac{2}{\sqrt{ \ln (3) n}} \right) \leq 12 \sqrt{ \ln (3) C_{SG} n}.
\end{align*}
where the second inequality follows from $m \geq \frac{2 \ln (3)}{C_{SG}} n$. Inserting this bound into \eqref{eq:Wm_hoelder} yields the claim with $C_W = 72 \sqrt{ \ln (3) C_{SG}}$.
\end{proof}

Now we are ready to apply Mendelson's small ball method \eqref{eq:mendelson}.
For $D$ defined in \eqref{eq:D} and measurements $\mat A_k = |\vec\alpha_k \rangle \! \langle \vec\alpha_k|$ with $\vec\alpha_k$ obeying \eqref{eq:sub_isotropy} and \eqref{eq:subexponential} the bounds from the previous Lemmas imply
\begin{align*}
  \frac{1}{\sqrt{m}}\inf_{\mat Z \in D} \|\mathcal{A}(\mat Z) \|_{\ell_1} \geq \xi \sqrt{m} C_Q \left( 1- \frac{4 \xi^2}{C_{SI}} \right)^2 - 2 C_W \sqrt{n} - \xi t \quad \forall \xi \in (0, 1/\sqrt{C_{SI}}), \forall t \geq 0
\end{align*}
with probability at least $1- \mathrm{e}^{-2t^2}$. We choose $\xi = \sqrt{C_{SI}}/4$ and $t = \gamma_1 \sqrt{m}$, where $\gamma_1 = \frac{9 C_Q}{32}$ and obtain with probability at least $1-\exp \left( -2 \gamma_1 m \right)$:
\begin{align*}
  \frac{1}{\sqrt{m}}\inf_{\mat Z \in D} \|\mathcal{A}(\mat Z) \|_{\ell_1} \geq & \frac{9 C_Q\sqrt{C_{SI}}}{64} \sqrt{m} -  C_W\sqrt{n} - \frac{\sqrt{C_{SI}}}{4} \frac{9 C_Q}{32} \sqrt{m} \\
  = & C_W \left( \frac{9 C_Q \sqrt{C_{SI}}}{128 C_W} \sqrt{m} - \sqrt{n} \right).
\end{align*}
Setting $m = C n$ with $C = \left( \frac{256 C_W}{9 C_Q \sqrt{C_l}} \right)^2$ implies
\begin{equation*}
  \frac{1}{\sqrt{m}} \inf_{\mat Z \in D} \| \mathcal{A}(\mat Z) \|_{\ell_1} \geq 2 C_W \sqrt{n} = \frac{2 C_W}{\sqrt{C}} \sqrt{m}
\end{equation*}
with probability at least $1- \mathrm{e}^{-2 \gamma_1 m}$.
For $\tau = \frac{ 2 C_W}{\sqrt{C}}$, the first claim in \cref{prop:RECR_nsp} follows from rearranging this expression and using $\| \mat Z \|_2=1$ for all $\mat Z \in D$.

Let us now move on to establishing the second statement \eqref{eq:approx_povm}:
Isotropy \eqref{eq:tight_frame} implies
\begin{align*}
  \frac{1}{ C_I m} \sum_{k=1}^m |\vec\alpha_k \rangle \! \langle \vec\alpha_k | - \mathbb{I}
  = \frac{1}{C_{SG} m} \sum_{k=1}^m \left( |\vec\alpha_k \rangle \! \langle \vec\alpha_k| - \mathbb{E} \left[ |\vec\alpha_k \rangle \! \langle \vec\alpha_k| \right] \right)
\end{align*}
and each $\alpha_k$ has subgaussian tails by assumption \eqref{eq:subexponential}.
Thus, \cref{thm:bernstein} is applicable and setting $t= \min \left\{\frac{1}{6},2 C_{SG} \right\}$ yields
\begin{align*}
  \mathrm{Pr} \left[ \left\| \frac{1}{C_I m} \sum_{k=1}^m |\vec\alpha_k \rangle \! \langle \vec\alpha_k| -  \mathbb{I} \right\|_\infty \geq \frac{1}{6} \right]
  \leq 2 \exp \left( 2 \ln (3) n - \frac{C_I m \min\left\{ 1/6, 2 C_{SG} \right\}}{8 C_{SG}} \right) \leq 2 \exp \left( - \gamma_2 m \right),
\end{align*}
where the second inequality follows from $m \geq C n$, provided that $C$ is sufficiently large. Finally, we use the union bound  for the overall probability of failure and set $\gamma := \min \left\{ 2 \gamma_1,\gamma_2 \right\}$.

\section{Proof of \cref{prop:gauss+recr_requirements}}
\label{sec:gauss+recr_requirements}

We now proof the crucial properties \eqref{eq:tight_frame}--\eqref{eq:subexponential} for the different measurement ensembles from \cref{prop:gauss+recr_requirements}.

\subsection{The Gaussian sampling scheme}

Let $\vec\alpha \in \mathbb{C}^n$ be a standard (complex) Gaussian vector and fix any $\vec z \in \mathbb{C}^n$.
Then, the random variable $\langle \vec\alpha,\vec z \rangle$ is an instance of a standard (complex normal) random variable $a = \tfrac{\| \vec z \|_{\ell_2}}{\sqrt{2}} \left(a_R + i a_I\right)$ with $a_R, a_I \sim \mathcal{N}(0,1)$.
In turn, $|a|^2 = \frac{\| \vec z \|_{\ell_2}^2}{2} (a_R^2 + a_I^2)$ is a re-scaled version of a $\chi^2$-distributed random variable with two degrees of freeom. The moments of such a random variable are well-known and we obtain
\begin{equation}
  \mathbb{E} (| \langle \vec\alpha,\vec z \rangle|^{2N})= \left( \frac{ \| \vec z \|_{\ell_2}}{\sqrt{2}}\right)^N \times 2^N N! = \| \vec z \|_{\ell_2}^N N! \; .\label{eq:moments_gauss}
\end{equation}
From this, we can readily infer $C_{SG} = 1$, and the special case $N=1$  yields $C_I=1$.

For the remaining expression, use an eigenvalue decomposition $\mat Z = \sum_{k=1}^d \zeta_k |\vec z^{(k)} \rangle \langle \vec z^{(k)}|$ (with normalized eigenvectors $\vec z^{(k)}\in \mathbb{C}^n$) and note that the random variables $|\langle \vec a,\vec z^{(1)} \rangle|,\ldots, | \langle \vec a,\vec z^{(n)} \rangle|$ are independently distributed and obey \cref{eq:moments_gauss}.
Consequently:
\begin{align*}
  \mathbb{E} \left[ \tr \left( \mat A \mat Z \right)^2 \right]
  =& \mathbb{E} \left[ \left( \sum_{k=1}^d \zeta_k | \langle \vec\alpha,\vec z^{(k)} \rangle|^2 \right)^2 \right] \\
  =& \sum_{k \neq l} \zeta_k \zeta_l \mathbb{E} \left[ |\langle \vec\alpha,\vec z^{(k)} \rangle|^2 \right] \mathbb{E} \left[ | \langle \vec a,\vec z^{(l)} \rangle|^2 \right]
  + \sum_{k=1}^d \zeta_k^2 \mathbb{E} \left[ | \langle \vec a, \vec z^{(k)} \rangle|^4 \right] \\
  =& \sum_{k \neq l} \zeta_k \zeta_l \|\vec z^{(k)} \|_{\ell_2}^2 \| \vec z^{(l)} \|_{\ell_2}^2 + 2 \sum_{k=1}^d \zeta_k^2 \| \vec z^{(k)} \|_{\ell_2}^4
  = \sum_{k,l=1}^d \zeta_k \zeta_l + 2\sum_{k=1}^d \zeta_k^2 \\
  =& \tr (\mat Z)^2 + \tr (\mat Z^2)
  \geq \| \mat Z \|_2^2,
\end{align*}
which implies $C_{SI} = 1$.

\subsection{The uniform sampling scheme}

Here, $\vec\alpha$ is chosen uniformly from the complex sphere with radius $\sqrt{n}$.
This in turn implies that the distribution of $\vec\alpha \in \mathbb{C}^n$ is invariant under arbitrary unitary transformations.
Techniques from representation theory -- more precisely: Schur's Lemma -- then imply
\begin{equation}
  \label{eq:from_schur}
  \mathbb{E} \left[ (|\vec\alpha \rangle \! \langle \vec\alpha| )^{\otimes N} \right] =
  n^N \binom{n+N-1}{N}^{-1} \mat P_{\vee^N},
\end{equation}
see e.g.\ \cite[Lemma~1]{scott_tight_2006}.
Here, $\mat P_{\vee^N}$, denotes the projector onto the totally symmetric subspace $\bigvee\!^N \subset \left( \mathbb{C}^n \right)^{\otimes N}$.
Note that $\left(|\vec z \rangle \! \langle \vec z| \right)^{\otimes N} \in \bigvee\!^N$ and, moreover $2 \mathrm{tr} \left( \mat P_{\vee^2} \mat Z^2 \right)= \| \mat Z \|_2^2 + \mathrm{tr} (\mat Z)^2$ for any matrix $\mat Z$, see e.g.\ \cite[Lemma~17]{kueng_low_2016}.
Consequently,
\begin{align*}
  \mathbb{E} \left[ | \langle\vec\alpha,\vec z \rangle|^2 \right]
  =& \mathrm{tr} \left( |\vec z \rangle \! \langle \vec z| \, \mathbb{E} \left[ |\vec\alpha \rangle \! \langle\vec\alpha| \right] \right)
  = \mathrm{tr} \left( |\vec z \rangle \! \langle \vec z| \mathbb{I} \right) = \| \vec z \|_{\ell_2}^2, \\
  \mathbb{E} \left[
  \langle\vec\alpha| Z |\vec\alpha \rangle^2 \right]
  =& \tr \left( \mathbb{E} \left[ (|\vec\alpha \rangle \! \langle\vec\alpha|)^{\otimes 2} \right] \mat Z^{\otimes 2} \right)
  = \frac{n}{n+1} \left( \| \mat Z \|_2^2 + \mathrm{tr}(\mat Z)^2 \right) \geq \frac{n}{n+1} \| \mat Z \|_2^2, \\
  \mathbb{E} \left[ | \langle\vec\alpha, \vec z \rangle |^{2N} \right]
  =& \mathrm{tr} \left(\mathbb{E} \left[ (|\vec\alpha \rangle \! \langle\vec\alpha|)^{\otimes N} \right]  (|\vec z\rangle \! \langle \vec z|)^{\otimes N}  \right)
  = n^N \binom{n+N-1}{N}^{-1} \| \vec z \|_{\ell_2}^{2N} \\
  =& N! \frac{n^N (n-1)!}{(n+N-1)!} \leq N!,
\end{align*}
which implies $C_I=1$, $C_{SI} = \frac{n}{n+1}$ and $C_{SG}=1$.

\subsection{The RECR sampling scheme}

\begin{lemma}[The RECR ensemble is isotropic on $\mathbb{C}^n$]
Suppose that $\vec\alpha$ is chosen from a RECR ensemble with erasure probability $1-p$. Then
\begin{align*}
  \mathbb{E} \left[ | \langle \vec \alpha,\vec z \rangle|^2 \right] = p \| \vec z \|_{\ell_2}^2
  \quad \forall \vec z \in \mathbb{C}^n.
\end{align*}
\end{lemma}

\begin{proof}
Let $\alpha_k = \langle \vec e_k, \vec\alpha\rangle$, where $\vec e_1,\ldots,\vec e_n$ is the orthonormal basis with respect to which the RECR vector is defined.
Theses components obey $\mathbb{E}\left[ \alpha_k \right] = \mathbb{E} \left[ \cc{\alpha}_k \right] = 0$, as well as $\mathbb{E} \left[ |\alpha_k|^2 \right] = p$.
For any $\vec z \in \mathbb{C}^n$ we then have
\begin{align*}
  \mathbb{E} \left[| \langle \vec \alpha, \vec z\rangle |^2 \right]
  =& \sum_{i,j=1}^n \mathbb{E} \left[ \cc{\alpha}_i \alpha_j \right] \langle \vec e_i | \vec z \rangle \langle \vec z | \vec e_j \rangle = p \sum_{i=1}^n | \langle \vec e_i, \vec z \rangle|^2 = p \| \vec z \|_{\ell_2}^2.
\end{align*}
\end{proof}

\begin{lemma}[The RECR ensemble is sub-isotropic on $\Hermitian_n$]
  \label{lem:recr_subisotropic}
  Suppose that $\vec\alpha$ is chosen from a RECR ensemble with erasure probability $1-p$. Then
  \begin{equation*}
  \mathbb{E} \left[ \langle \vec \alpha| \mat Z |\vec \alpha \rangle^2 \right] \geq p \min \left\{ p, 1-p \right\} \| \mat Z \|_2^2 \quad \forall \mat Z \in \Hermitian_n
  \end{equation*}
\end{lemma}

\begin{proof}
Fix $\mat Z \in \Hermitian_n$ and compute
\begin{align*}
  \mathbb{E} \left[ \langle \vec\alpha | \mat Z | \vec\alpha \rangle^2 \right]
  =& \sum_{i,j,k,l} \mathbb{E} \left[ \bar{\alpha}_i \alpha_j \cc{\alpha^\prime_k} \alpha^\prime_l \right] \langle \vec e_i|\mat Z| \vec e_j \rangle \langle \vec e_k |\mat Z| \vec e_l \rangle \\
  =& \sum_{i} \mathbb{E} \left[ | \alpha_i |^4 \right] \langle \vec e_i|\mat Z|\vec e_i \rangle^2 + \sum_{i \neq k} \mathbb{E} \left[ | \alpha_i |^2 | \alpha_k|^2 \right] \left( \langle \vec e_i|\mat Z|\vec e_i \rangle \langle \vec e_k|\mat Z|\vec e_k \rangle + \langle \vec e_i|\mat Z|\vec e_k \rangle \langle \vec e_k| \mat Z|\vec e_i \rangle \right) \\
  =& p \sum_{i=1}^n \langle \vec e_i|\mat Z|\vec e_i \rangle^2 + p^2 \sum_{i \neq k} \left( \langle \vec e_i|\mat Z|\vec e_i \rangle \langle \vec e_k|\mat Z|\vec e_k \rangle + \langle \vec e_i|\mat Z|\vec e_k \rangle \langle \vec e_k |\mat Z| \vec e_i\rangle \right) \\
  =& p^2 \sum_{i,k=1}^n \left( \langle \vec e_i|\mat Z|\vec e_i \rangle \langle \vec e_k|\mat Z|\vec e_k \rangle + \langle \vec e_i|\mat Z|\vec e_k \rangle \langle \vec e_k |\mat Z| \vec e_i\rangle \right) + p(1-2 p) \sum_{i=1}^n \langle \vec e_i|\mat Z|\vec e_i \rangle^2 \\
  =& p^2 \left( \tr (\mat Z)^2 + \|\mat Z\|_2^2 \right) + p (1-2 p) \sum_{i=1}^n \langle \vec e_i| \mat Z | \vec e_i \rangle^2 \\
  \geq& p^2 \|\mat Z\|_2^2 + p(1-p) \sum_{i=1}^n \langle \vec e_i |\mat Z|\vec e_i \rangle^2
\end{align*}
Finally, we make a case distinction:
\begin{itemize}
\item[$p \leq 1/2$]: This implies $p(1-2p) \geq 0$ and consequently
\begin{align*}
  \mathbb{E} \left[ \langle \vec\alpha |\mat Z| \vec\alpha \rangle^2 \right] \geq p^2 \| \mat Z \|_2^2.
\end{align*}
\item[$p \geq 1/2$]: Use $\sum_{i=1}^n \langle i| X|i \rangle^2 \leq \| X \|_2^2$ to conclude
\begin{align*}
  \mathbb{E} \left[ \langle \vec\alpha | \mat Z | \vec\alpha \rangle^2 \right]
\geq ( p^2 - p|1-2p|) \|\mat Z\|_2^2 = p(1-p) \| \mat Z \|_2^2.
\end{align*}
\end{itemize}
\end{proof}

\begin{lemma}[Subgaussian tails of the RECR distribution]
Suppose that $\vec\alpha$ is a vector from the RECR ensemble. Then
\begin{align*}
\mathbb{E} \left[ | \langle \vec\alpha, \vec z \rangle|^{2N} \right] \leq \mathrm{e}^{\frac{3}{2}} N! \quad \forall \vec z \in \Sphere^{n-1}.
\end{align*}
\end{lemma}

\begin{proof}
Fix $\vec z \in \mathbb{C}^n$ with $\| \vec z \|_{\ell_2}=1$ and note that $|\alpha_k| \leq 1$ together with the independence of $\alpha_k,\alpha_l$ for $k \neq l$ implies
\begin{align}
  \mathbb{E} \left[ \exp \left( | \langle \vec\alpha, \vec z \rangle|^2 \right) \right]
  =& \mathbb{E} \left[ \prod_{k=1}^n \exp \left( | \alpha_k|^2 |z_k|^2 \right) \prod_{k \neq l} \exp \left( \cc{\alpha}_k \alpha_l \cc{z}_k z_l \right) \right] \nonumber \\
  \leq& \exp \left( \| \vec z \|_{\ell_2}^2 \right) \prod_{k \neq l} \mathbb{E} \left[  \exp \left( \cc{\alpha}_k \alpha_l \cc{z}_k z_l \right)  \right]. \label{eq:moment_aux1}
\end{align}
Now note that for $k \neq l$, $\cc{\alpha}_k \alpha_l$ is again a RECR random variable $\tilde{\alpha}_{k,l}$, but with erasure probability $1-p^2$.
Moreover, every RECR random variable $\alpha$ can be decomposed into the product of two independent random variables: $ \alpha= \eta \omega$, where $\eta$ is a Rademacher random variable and $\omega \in \left\{0, 1,i \right\}$ obeys $| \omega | \leq 1$.
Consequently
\begin{align*}
  \mathbb{E} \left[ \exp \left( \bar{\alpha}_k \alpha_l \bar{z}_k z_l \right) \right]
  =& \mathbb{E} \left[ \exp \left( \tilde{\alpha}_{k,l} \bar{z}_k z_l \right) \right]
  = \mathbb{E}_{\omega} \left[ \mathbb{E}_\eta \left[ \eta \omega \bar{z}_k  z_l \right] \right]
  = \mathbb{E}_{\omega} \left[ \cosh \left( \omega \bar{z}_k z_l \right) \right] \\
  \leq & \mathbb{E}_\omega \left[ \exp \left( |\omega \bar{z}_k z_l|^2/2 \right) \right]
  \leq  \exp \left( \frac{|z_k|^2 |z_l|^2}{2} \right),
\end{align*}
where we have used the standard estimate $\cosh (x) \leq \exp \left( |x|^2/2 \right)$ $\forall x \in \mathbb{C}$, as well as $| \omega| \leq 1$. Inserting this bound into \eqref{eq:moment_aux1} yields
\begin{align*}
  \mathbb{E} \left[ \exp \left( | \langle \vec \alpha, \vec z \rangle|^2 \right) \right]
  \leq \exp \left( \| \vec z \|_2^2 \right) \prod_{k \neq l} \exp \left( \frac{|z_k|^2 |z_l|^2}{2} \right)
  \leq \exp \left( \| \vec z \|_2^2 + \frac{1}{2}\| \vec z \|_{\ell_2}^4 \right) = \mathrm{e}^{\frac{3}{2}},
\end{align*}
because $\| \vec z \|_{\ell_2}=1$.
Markov's inequality shows that this exponential bound implies a subexponential tail bound for the random variable $| \langle \vec \alpha,\vec z \rangle|^2$:
\begin{align*}
  \mathrm{Pr} \left[ | \langle \vec\alpha,\vec z \rangle|^2 \geq t \right]
  =& \mathrm{Pr} \left[ \exp \left( | \langle \vec \alpha,\vec z \rangle|^2 \right) \geq \exp \left( t \right) \right]
  \leq \frac{ \mathbb{E} \left[ \exp \left( | \langle \vec\alpha, \vec z \rangle|^2 \right) \right]}{\exp (t)} \leq \mathrm{e}^{\frac{3}{2}-t}.
\end{align*}
This in turn implies the following bound on the moments:
\begin{align*}
  \mathbb{E} \left[ | \langle \vec \alpha,\vec z \rangle|^{2N} \right]
  =  N \int_0^\infty \mathrm{Pr}\left[ | \langle \vec \alpha,\vec z \rangle|^2\geq t \right] t^{N-1} \mathrm{d}t \leq N \mathrm{e}^{\frac{3}{2}} \int_0^\infty \mathrm{e}^{-t} t^{N-1} \mathrm{d}t
  = \mathrm{e}^{\frac{3}{2}} N!,
\end{align*}
where we have used a well-known integration formula for moments, see e.g.\ \cite[Prop.~7.1]{Foucart_2013_Mathematical}, as well as integration by parts.
\end{proof}

\section{The normalised RECR scheme}
\label{sec:normalized_recr}

We have seen that the unnormalized RECR measurement ensemble obeys all conditions necessary for establishing strong PhaseLift recovery guarantees, most notably \cref{thm:phaselift_noisy}.
In this section, we shift our attention to the \emph{normalized} RECR ensemble instead. I.e.\ each measurement $\tilde{\mat{A}}_k = \tfrac{n}{\norm{\mat{A}_k}_2}\mat{A}_k$ with $\mat{A}_k = \ket{\vec{\alpha}_k} \bra{\vec{\alpha}_k}$ and $\ket{\vec{\alpha}_k} \in \left\{0,\pm 1,\pm i \right\}^n$ is an outer product of RECR vectors renormalized to length $\norm{\mat{A}_k} =n$. 
For the unnormalized RECR ensemble, we have seen 
 that w.h.p.\ any $\mat X = |\vec x \rangle \!\langle \vec x|$ can be recovered from $m \geq  C n$ measurements of the form
\begin{equation*}
  y_k = \tr \left( \mat A_k \mat X \right) + \epsilon_k
\end{equation*}
via solving
\begin{equation}
  \underset{\mat Z\geq 0}{\textrm{minimize}} \quad \| \mathcal{A}(\mat Z) - \vec y \|_{\ell_1}. \label{eq:phaselift2}
\end{equation}
The solution $\mat{Z}^\sharp$ of this program is guaranteed to obey
\begin{align*}
\| \mat{Z}^\sharp - \mat X \|_2 \leq \frac{C' \| \vec\epsilon \|_{\ell_1}}{m}.
\end{align*}
Now suppose that we have $m$ normalized RECR measurements instead: $\tilde{\mat A}_k = \frac{n}{\| \mat A_k \|_{2}} \mat A_k$. Then the associated measurements correspond to
\begin{equation*}
  \tilde{y}_k = \tr \left( \tilde{\mat A}_k X \right) + \epsilon_k = \frac{ n}{\| \mat A_k \|_2} \tr \left( \mat A_k X \right) + \tilde{\epsilon}_k.
\end{equation*}
Multiplying this expression by $\frac{\| \mat A_k \|_2}{n}$ yields
\begin{equation*}
\underset{:= y_k}{\underbrace{ \frac{ \| \mat A_k \|_2}{n}\tilde{y}_k}}
= \tr \left( \mat A_k X \right) + \underset{:= \epsilon_k}{\underbrace{ \frac{ \| \mat A_k \|_2}{n}\tilde{\epsilon}_k}}
\end{equation*}
and solving \eqref{eq:phaselift2} for re-scaled measurement outcomes $y_k = \frac{ \| \mat A_k \|_2}{n}\tilde{y}_k$ yields an estimator of $\mat X$ that is guaranteed to obey
\begin{equation*}
  \| \mat{Z}^\sharp - \mat X \|_2 \leq \frac{C' \| \vec \epsilon \|_{\ell_1}}{m}
  = \frac{C'}{m} \sum_{k=1}^m \frac{ \| \mat A_k \|_2}{n} | \tilde{\epsilon}_k |
  \leq \frac{C'}{m} \sum_{k=1}^m | \tilde{\epsilon}_k| = \frac{C' \| \tilde{\vec\epsilon} \|_{\ell_1}}{m}.
\end{equation*}
Here, the last line is due to $\| \mat A_k \|_2 = \| \vec\alpha_k \|_{\ell_2}^2 \leq n$.
We can conclude that small modifications in the PhaseLift algorithm ensure that normalized RECR measurements perform at least as well as unnormalized RECR measurements. However, we already know that unnormalized RECR measurements are accompanied by strong theoretical convergence guarantees, namely \cref{thm:phaselift_noisy}.

\end{document}